\documentclass[runningheads]{llncs}

  \makeatletter
\renewcommand\section{\@startsection
  {section}{1}{\z@}%
  {-10pt plus -4pt minus -2pt}
  {10pt plus 2pt minus 2pt}
  {\normalfont\large\bfseries}%
}

\renewcommand\subsection{\@startsection
  {subsection}{2}{\z@}%
  {-8pt plus -3pt minus -2pt}%
  {6pt plus 2pt minus 2pt}%
  {\normalfont\normalsize\bfseries}%
}
\makeatother

\PassOptionsToPackage{table}{xcolor}

\usepackage[table]{xcolor}
\usepackage{booktabs}
\usepackage{mathtools}
\usepackage{hyperref}
\hypersetup{
  colorlinks,
  linkcolor={black},
  citecolor={black},
  urlcolor={blue!80!black}
}

\usepackage{graphicx}
\usepackage{afterpage}
\usepackage{wrapfig}
\usepackage{tcolorbox}
\usepackage{caption}
\definecolor{lightblue}{rgb}{0.68, 0.85, 0.9}
\newtcbox{\mybox}{on line,
  boxrule=0pt,colframe=lightblue!75!white,colback=lightblue!75!white,
  boxsep=0pt,left=2pt,right=2pt,top=2pt,bottom=2pt,arc=6pt}

\usepackage{amsmath}

\usepackage{amssymb}
\usepackage{mathtools}
\usepackage{proof}
\usepackage{stmaryrd}
\usepackage{pifont}
\usepackage{marvosym}
\usepackage{xspace}
\SetSymbolFont{stmry}{bold}{U}{stmry}{m}{n} 
\usepackage{wrapfig}
\usepackage[bibliography=common]{apxproof}

\let\emptyset\varnothing

\usepackage{listings}

\lstdefinelanguage{program}{
  keywords={
    action,procedure,proc,invariant,inv,function,
    returns,global,refines,
    datatype,
    const,var,
    assume,assert,requires,ensures,preserves,
    if,then,else,call,exec,match, seq, reduce, par, right, left, atomic, axiom
  },
  morecomment=[l]{//},
  morecomment=[s]{/*}{*/},
  morecomment=[n]{(**}{**)},
  mathescape=true,
  escapeinside={(*@}{@*)},
  literate={<|}{$\langle$}1 {|>}{$\rangle$}1,
}

\definecolor{commentgreen}{RGB}{2,112,10}

\lstdefinestyle{lnumbers}{
  basicstyle = \ttfamily\scriptsize,
  keywordstyle     = \bfseries,
  commentstyle     = \color{commentgreen},
  emph             = {linear,linear_in,linear_out},
  emphstyle        = \textit,
  columns          = fullflexible,
  keepspaces       = true,
  numbers          = left,
  numbersep        = 3pt,
  numberstyle      = \ttfamily\tiny\color{gray},
  numberblanklines = false,
  captionpos       = b,
    numbers=none,
  linewidth=\linewidth,
  columns=fullflexible,
  breaklines=true,
  breakatwhitespace=true,
  aboveskip=2pt, belowskip=2pt,
  xleftmargin=2pt,
  xrightmargin=2pt
}

\makeatletter
\lst@AddToHook{OnEmptyLine}{\vspace{-0.6\baselineskip}}
\makeatother

\makeatletter
\lst@AddToHook{OnEmptyLine}{\addtocounter{lstnumber}{-1}}
\makeatother

\mathchardef\mhyphen="2D

\newcommand{\Val}{\mathit{Val}}
\newcommand{\Var}{\mathit{Var}}
\newcommand{\GlobalVar}{\mathit{GVar}}
\newcommand{\LocalVar}{\mathit{LVar}}

\newcommand{\Store}{\mathit{Store}}
\newcommand{\GlobalStore}{\mathit{GStore}}
\newcommand{\LocalStore}{\mathit{LStore}}

\newcommand{\Gate}{\mathit{Gate}}
\newcommand{\Trans}{\mathit{Trans}}

\newcommand{\Stmt}{\mathit{Stmt}}

\newcommand{\Action}{\mathit{Action}}

\newcommand{\ProcSignature}{\mathit{ProcSig}}

\newcommand{\ActionName}{\mathit{ActionName}}
\newcommand{\ProcName}{\mathit{ProcName}}

\newcommand{\IOMap}{\mathit{IOMap}}

\newcommand{\Prog}{\mathit{Prog}}
\newcommand{\MoverType}{\mathit{MoverType}}

\newcommand{\gate}{\rho}
\newcommand{\trans}{\tau}

\newcommand{\var}{v}
\newcommand{\lvar}{x}
\newcommand{\OutVar}{O}
\newcommand{\InVar}{I}
\newcommand{\MoverVar}{M}
\newcommand{\mayFailVar}{F}

\newcommand{\store}{\sigma}
\newcommand{\globalStore}{g}
\newcommand{\localStore}{\ell}

\newcommand{\stmt}{s}

\newcommand{\action}{A}
\newcommand{\proc}{Q}
\newcommand{\procOrAction}{X}

\newcommand{\callee}{Q}

\newcommand{\imap}{\iota}
\newcommand{\omap}{o}

\newcommand{\prog}{\mathcal{P}}
\newcommand{\sourceprog}{\mathcal{P}_s}

\newcommand{\interprog}{\mathcal{P}_i}
\newcommand{\reducedprog}{\mathcal{P}_r}

\newcommand{\commutes}{\mathit{commutes}}
\newcommand{\preservesSuccess}{\mathit{preserves \mhyphen success}}
\newcommand{\preservesFailure}{\mathit{preserves \mhyphen failure}}

\newcommand{\uncommitted}{\mathsf{uncommitted}}
\newcommand{\committed}{\mathsf{committed}}
\newcommand{\commit}[1]{\mathsf{block}(#1).\mathsf{commit}} 
\newcommand{\block}{\mathsf{block}}

\newcommand{\ps}{\mathit{ps}}
\newcommand{\as}{\mathit{as}}

\newcommand{\stmtfont}[1]{\mathtt{#1}}
\newcommand{\skipstmt}{\stmtfont{skip}}
\newcommand{\ifstmt}[3]{\stmtfont{if} \ #1 \ #2 \ #3}

\newcommand{\callstmt}[2][]{\stmtfont{call}_{#1}\: #2}

\newcommand{\callstmtempty}{\stmtfont{call}}

\newcommand{\seqstmt}[2]{#1\, \stmtfont{;}\, #2}
\newcommand{\atomicstmt}[1]{\stmtfont{atomic}\: #1}

\newcommand{\inatomicstmt}[1]{\stmtfont{in \mhyphen atomic}\: #1}
\newcommand{\inatomicstmtempty}{\stmtfont{in \mhyphen atomic}}
\newcommand{\parstmt}[2]{#1\ \stmtfont{par}\ #2}
\newcommand{\parstmtempty}{\stmtfont{par}}

\newcommand{\parreduce}[2]{\stmtfont{par \mhyphen reduce}\: #1\ \stmtfont{par}\ #2}
\newcommand{\parreduceempty}{\stmtfont{par \mhyphen reduce}}

\newcommand{\seqreduce}[1]{\stmtfont{seq \mhyphen reduce}\: #1}
\newcommand{\seqreduceempty}{\stmtfont{seq \mhyphen reduce}}
\newcommand{\inseqreduce}[1]{\stmtfont{in \mhyphen seq \mhyphen reduce}\: #1}
\newcommand{\inseqreduceempty}{\stmtfont{in \mhyphen seq \mhyphen reduce}}

\newcommand{\refines}{\preccurlyeq}

\newcommand{\movertype}{\mathit{mover \mhyphen type}}
\newcommand{\leftmover}{\mathbf{L}}
\newcommand{\rightmover}{\mathbf{R}}
\newcommand{\nonmover}{\mathbf{N}}
\newcommand{\bothmover}{\mathbf{B}}

\newcommand{\welltyped}{\mathit{well \mhyphen typed}}

\newcommand{\mayfail}{\mathit{may \mhyphen fail}}
\newcommand{\localvarset}{\mathit{l}}

\newcommand{\step}{\rightarrow}

\newcommand{\stepp}[1][]{\xrightarrow{#1}}
\newcommand{\transitive}{\mathrel{\vphantom{\to}^{*}}}
\newcommand{\fail}{\lightning}

\newcommand{\stmtCtxt}{\mathit{SC}}
\newcommand{\threadCtxt}{\mathit{TC}}
\newcommand{\poolCtxt}{\mathit{PC}}
\newcommand{\leafCtxt}{\mathit{LC}}
\newcommand{\hole}{\bullet}

\newcommand{\Tree}{t}
\newcommand{\Leaf}[1]{\mathsf{Lf}\: #1}
\newcommand{\Node}[2]{\mathsf{Nd}\: #1\ #2}
\newcommand{\EmptyNode}{\mathsf{Nd}}
\newcommand{\EmptyLeaf}{\mathsf{Lf}}
\newcommand{\ThreadPool}{\mathcal{T}}
\newcommand{\conf}{c}

\newcommand{\undefbio}{\text{\Biohazard}}

\newcommand{\modified}{\mathit{mod}}

\newcommand{\inputbinding}{\mathit{in \mhyphen bind}}

\newcommand{\seq}{\mathit{seq}}

\newcommand{\labeling}{\mathit{label}}

\newcommand{\code}[1]{\texttt{\small#1}}

\newcommand{\lang}{\textsf{RedPL}\xspace}
\newcommand{\civl}{Civl\xspace}

\newcommand{\pto}{\rightharpoonup}
\newcommand*{\cc}{\kern-.2em\cdot\kern-.2em}

\newcommand{\blank}{\_}

\newcommand{\ov}[1]{\overline{#1}}
\newcommand{\set}[1]{\{#1\}}

\DeclareMathOperator{\dom}{dom}
\DeclareMathOperator{\img}{img}
\newcommand{\inverse}[1]{#1^{-1}}

\newcommand{\true}{\mathit{true}}
\newcommand{\false}{\mathit{false}}

\newcommand{\circnum}[1]{\ding{\the\numexpr 171 + #1 \relax}}
\newcommand{\Circnum}[1]{\ding{\the\numexpr 181 + #1 \relax}}

\newcommand{\condition}[1]{\normalfont\textbf{(\texttt{#1})}}
\newcommand{\wlp}{\mathit{wlp}}
\newcommand{\update}[1]{{#1}}
\newcommand{\edit}[1]{\textcolor{black}{#1}}

\begin{document}

\title{Reduction for Structured Concurrent Programs}

\author{Namratha Gangamreddypalli\inst{1} \and
Constantin Enea\inst{1}\and Shaz Qadeer\inst{2}}%
\authorrunning{Namratha Gangamreddypalli \and
Constantin Enea \and Shaz Qadeer}

\institute{LIX, Ecole Polytechnique, CNRS and Institut Polytechnique de Paris, France \\ \email{\{namratha, cenea\}@lix.polytechnique.fr}
\and Microsoft \email{shaz.qadeer@gmail.com}}%

\maketitle
\setcounter{footnote}{0}
\renewcommand{\thefootnote}{\arabic{footnote}}

\begin{abstract}
Commutativity reasoning based on Lipton's movers is a powerful technique for verification of concurrent programs.
The idea is to define a program transformation that preserves a subset of the initial set of interleavings,
which is sound modulo reorderings of commutative actions.
Scaling commutativity reasoning to routinely-used features in software systems,
such as procedures and parallel composition, remains a significant challenge.

In this work, we introduce a novel reduction technique for structured concurrent programs that unifies two key advances.
First, we present a reduction strategy that soundly replaces \edit{parallel composition} with sequential composition.
Second, we generalize Lipton’s reduction to support atomic sections containing (potentially recursive) procedure calls.
Crucially, these two foundational strategies can be composed arbitrarily, greatly expanding the scope and flexibility of reduction-based reasoning.
We implemented this technique in Civl and demonstrated its effectiveness on a number of challenging case studies,
including a snapshot object, a fault-tolerant and linearizable register, the FLASH cache coherence protocol, and a non-trivial variant of Two-Phase Commit.
\end{abstract}

\section{Introduction}

Commutativity reasoning is a powerful technique for verification of concurrent programs.
This method derives from the observation that certain pairs of concurrently-executing statements can be reordered without affecting program behavior,
i.e., such statements commute.
Interleavings (i.e., concurrent execution sequences) that differ only in the ordering of commuting statements are considered equivalent.
As a result, it is sufficient to verify the correctness of a single representative interleaving from each equivalence class.
These techniques are also called \emph{reduction} techniques because they reduce reasoning to a smaller set of representative interleavings.

In static verification of concurrent programs,
a standard method to exploit commutativity reasoning is to capture the reduced program via a syntactic
transformation of the original program.
The proof of correctness is then done on the transformed program and the reduction argument is used to carry over the results of the
verification to the original program.
Lipton~\cite{DBLP:journals/cacm/Lipton75} introduced atomic sections as a simple method to capture such a transformation.
Since then, atomic sections have been used extensively to specify non-interference in and simplify reasoning
about concurrent programs~\cite{DBLP:conf/pldi/FlanaganQ03,DBLP:journals/tse/FlanaganFQ05,DBLP:conf/popl/ElmasQT09,DBLP:conf/cav/HawblitzelPQT15,DBLP:journals/pacmpl/FlanaganF20}.

In this paper, we focus on the application of commutativity reasoning towards deductive verification of concurrent programs.
Applying commutativity reasoning to real-world programs is challenging.
Software systems routinely use procedures for code structuring and scaling software engineering.
Concurrent systems, in addition, are performance oriented and often launch multiple tasks in parallel, collecting results as the tasks complete.
These features, procedures and dynamic concurrency, are not adequately addressed by existing approaches.
Lipton~\cite{DBLP:journals/cacm/Lipton75} only addresses the problem of concurrent programs with bounded threads 
and atomic sections of bounded size.
QED~\cite{DBLP:conf/popl/ElmasQT09} and Civl~\cite{DBLP:conf/cav/HawblitzelPQT15}
handle atomic sections containing loops, but do not handle procedure calls (unless they are inlined beforehand) or parallel composition.
Kragl et al.~\cite{DBLP:conf/concur/KraglQH18} present a program transformation that synchronizes asynchronous procedure calls
by demonstrating that the called procedure can be summarized as a single atomic action that commutes to the left of any other
program action (a so-called left mover).
While asynchronous calls can be viewed as a restricted form of parallel composition,
this model is not general enough for our purposes (see \autoref{sec:related}).

We introduce a novel reduction technique for structured concurrent programs that unifies two key advances.
First, we present a reduction strategy that soundly replaces \edit{parallel composition} with sequential composition,
addressing a dimension orthogonal to atomic section introduction explored in prior work.
Second, we generalize Lipton’s reduction to support atomic sections containing (potentially recursive) procedure calls.
Crucially, these two foundational strategies can be composed arbitrarily, greatly expanding the scope and flexibility of reduction-based reasoning. 

Our reduction technique is based on a concept of \emph{movers} or
commuting statements~\cite{DBLP:journals/cacm/Lipton75,DBLP:conf/popl/ElmasQT09,DBLP:conf/cav/HawblitzelPQT15}.
The soundness of atomic section introduction relies on demonstrating that it consists of a sequence of \emph{right} movers,
followed by an arbitrary statement, and then a sequence of \emph{left} movers.
Right movers are statements that commute to the right of any other statement in the program,
while left movers commute to the left;
this terminology imagines time flowing from left to right.
\edit{Importantly, mover classification is relative to the set of actions under consideration: whether a statement is a right or left mover depends on how it interacts with the other actions in the program.}
This straightforward intuition for movers is deceptive;
precise definitions are non-trivial since they must account for statements that may fail or are non-deterministic. 

To handle structured code, we extend the notion of movers to procedures through a type system~\cite{DBLP:conf/pldi/FlanaganQ03}
that analyzes their bodies.
This enables the sound introduction of atomic sections that may include recursive procedure calls
(as before, soundness relies on ensuring a well-structured sequencing of right and left movers).
For example, if the body of a procedure $Q$ consists solely of right mover statements--including nested procedure calls,
which are recursively typed--then $Q$ is classified as a right mover procedure.
An analogous classification applies to left mover procedures, with the additional requirement that they must terminate when executed in isolation.
This termination condition ensures the preservation of failure behaviors:
a non-terminating procedure could otherwise unsoundly eliminate potential failures.
Notably, although we are reasoning about concurrency, this condition relies solely on the behavior of the procedure when
executed in isolation, a surprising and useful aspect of our framework.

The second key contribution of our technique is the ability to soundly transform a \edit{parallel}
construct, where an arbitrary number of threads are spawned and joined immediately afterward,
into a sequential composition of their respective code blocks.
This transformation again builds on the notion of movers and is achieved through an iterative process
that sequences left-mover code blocks first and right-mover code blocks last. 

The two contributions of our technique are integrated within a unified framework that supports arbitrary 
\edit{sequential} and \edit{parallel} composition of procedures.
Furthermore, the two contributions work in tandem:
transforming \edit{parallel} constructs may yield left or right mover procedures which
may further enable the introduction of atomic sections. 
Our technique is formalized in a core programming language,
where the two reduction principles are invoked via specific keywords,
and a type system guarantees their correct and sound composition.
\edit{We note that our notion of movers is more general than that used in Civl;
that is, certain statements commute under our definition but not under Civl's.
Moreover, our work provides the first formal proof of the soundness of Civl's reduction theory.
This proof played a key role in identifying the most general form of commutativity sufficient to ensure soundness (\autoref{sec:movers}).}

\update{
The soundness of this reduction framework is based on non-trivial arguments. For instance, reductions are defined at code level, which means that even a single reduction can apply an unbounded number of times in an execution (where the same procedure is called multiple times). This requires defining a non-trivial strategy for reordering steps in an execution wrt their commutativity properties. Also, the side conditions for using left or right movers are asymmetric, which may seem counterintuitive. For instance, in parallel reduction, right movers are not allowed to fail, and left mover procedures, when run without interference, are required to terminate—the latter is notable because it concerns sequential executions, even though it is applied to concurrent programs.} \edit{Furthermore, the restriction that right movers must be non-failing applies only to parallel reduction and not to sequential reduction.}

We have implemented our technique as an extension to Civl, preserving compatibility with its existing features.
To assess its effectiveness, we applied our implementation to a series of challenging case studies:
a parallel implementation of a snapshot object~\cite{DBLP:journals/jacm/AfekADGMS93}, 
the ABD register~\cite{DBLP:journals/jacm/AttiyaBD95} which simulates shared memory over message passing, 
the FLASH cache coherence protocol~\cite{DBLP:conf/isca/KuskinOHHSGCNBHGRH94},
and a non-trivial variant of the Two-Phase Commit protocol. \update{These examples span diverse domains, including concurrent objects, distributed protocols, and hardware cache coherence, demonstrating the broad applicability of our approach.}
\edit{In particular, the first two case studies are concurrent objects for which we prove that they are linearizable~\cite{DBLP:journals/toplas/HerlihyW90}. Proving linearizability for these objects is known to be challenging, because it can not be done via so-called fixed linearization points, i.e., the effect of a method invocation cannot be mapped to the execution of a fixed statement in the body of the method, and it requires prophecy variables\cite{DBLP:journals/tcs/AbadiL91}.}

Reduction was indispensable for our case studies,
each of which involves fine-grained access to shared state by an unbounded number
and dynamically-created concurrent tasks.
\update{The previous version of Civl could not handle these case studies because reduction was applicable only to sequential code fragments consisting solely of actions (and no procedure calls) and had no notion of parallel reduction.}
We are not aware of any other proof technique based on reduction that can handle our case studies.
\edit{Without reduction, proofs based purely on inductive invariants would be substantially more complex: the required invariants would be large and difficult to formulate.}


\section{Overview}\label{sec:overview}
\vspace{-2mm}
\begin{figure}[t]
\centering

\begin{minipage}[t]{0.58\textwidth}
\begin{lstlisting}
procedure scan() returns (snapshot: [int]StampVal) {
  var r1: [int]StampVal;
  var r2: [int]StampVal;
  while (true) {
    (call r1[1] := read(1)) par
    (call r1[2] := read(2)); 
    (call r2[1] := read(1)) par
    (call r2[2] := read(2));
    if (r1 == r2) {
      snapshot := r1;
      return;
    }   
  }
}
\end{lstlisting}
\end{minipage}\hfill
\begin{minipage}[t]{0.40\textwidth}
\begin{lstlisting}
datatype StampVal {
  StampVal(ts: int, value: Value)
}

var mem: [int]StampVal;

action read (i: int) returns (v: StampVal) {
  v := mem[i];
}

action write(i: int, v: Value) {
  mem[i] := StampVal(mem[i]->ts + 1, v);
}

action scan_spec() returns (snapshot: [int]StampVal) {
  assume (snapshot := mem);
}
\end{lstlisting}
\end{minipage}

\vspace{-2mm}
\caption{A snapshot object. The scan procedure carries out two consecutive collects, meaning it reads the entire memory in parallel twice. If both collects yield identical results, the procedure returns. Otherwise, it restarts.}
\label{fig:snapshot}
\vspace{-2mm}
\end{figure}

\begin{figure}[t]
\centering

\begin{minipage}[t]{0.58\textwidth}
\begin{lstlisting}
procedure scan() returns (snapshot: [int]StampVal) {
  var r1: [int]StampVal;
  var r2: [int]StampVal;
  while (true) {
    seq-reduce {
      par-reduce {
        (call r1[1] := read_f(1)) par
        (call r1[2] := read_f(2))
      }
      par-reduce {
        (call r2[1] := read_s(1)) par
        (call r2[2] := read_s(2))
      }
      if (r1 == r2) {
        snapshot := r1;
        return;
      }
    }
  }
}
\end{lstlisting}
\end{minipage}\hfill
\begin{minipage}[t]{0.40\textwidth}
\begin{lstlisting}
right action read_f(i: int) returns (out: StampVal) {
  var k: int;
  var v: Value;
  if (*) {
    assume k < mem[i]->ts; 
    out := StampVal(v, k);
  } else {
    out := mem[i];
  }
}

left action read_s(i: int) returns (out: StampVal) {
  var k: int;
  var v: Value;
  if (*) {
    assume k > mem[i]->ts; 
    out := StampVal(v, k);
  } else {
    out := mem[i];}
}
\end{lstlisting}
\end{minipage}

\vspace{-2mm}
\caption{An abstraction \code{scan}. Compared to the original, the two memory reads call the abstracted actions \code{read\_f} and \code{read\_s}, resp. In these actions, $*$ is non-deterministic choice and local variables are initially assigned arbitrary values. The annotations \code{seq-reduce} and \code{par-reduce} are related to our reduction technique.}
\label{fig:snapshot:abs}
\vspace{-6mm}
\end{figure}

We demonstrate our reduction proof technique on an implementation of a concurrent \emph{snapshot} object~\cite{DBLP:journals/jacm/AfekADGMS93} that provides two methods: \code{write(i,v)} that writes value \code{v} to memory cell \code{i}, and \code{scan()} which returns a snapshot of the entire memory. We assume that the memory is represented using an array. These methods can be called concurrently from an arbitrary number of threads. We first describe the implementation and the specification we are trying to prove, and then detail the application of our reduction proof technique.

\subsection{A Concurrent Snapshot Object}\label{ssec:overview:impl}

\noindent
\textbf{Implementation.} \autoref{fig:snapshot} lists the code of the snapshot object (\code{scan\_spec} is explained later). Each memory cell holds a timestamped value (a value along with an integer timestamp). For simplicity, we consider a memory with just two cells. The arbitrary-size case is considered in \autoref{sec:overview:unbounded}.

The code uses a programming language with regular procedure calls and parallel composition, where each access to the shared memory is encapsulated into a so-called \emph{action}. Actions are assumed to execute atomically in a single indivisible step. In this case, we have two actions, \code{read(i)} reads the \code{i}-th memory cell, and \code{write(i,v)} updates the \code{i}-th memory cell with value \code{v} and a timestamp incremented by 1 from its current timestamp.

The procedure \code{scan} consists of a ``spin'' loop that exits when two consecutive reads of the entire memory yield identical results. A read of the memory is parallelized, each memory cell is read in a different thread. This is written using the \code{par} keyword in between the two calls to the action \code{read} (actions are called in the same way as procedures). The meaning of $s\ \code{par}\ s'$ for two statements $s$ and $s'$ is that $s$ and $s'$ are executed in two different threads which are joined before executing the next statement. This ensures that the two reads of the entire memory do not overlap in time. For an expert reader, this corresponds to the fork-join model.

The procedure \code{write(i,v)} is omitted; it simply calls the homonymous action.

\smallskip
\noindent
\textbf{Specification: Linearizability.} Our goal is to show that this object is \emph{linearizable}, i.e., each concurrent invocation of \code{scan} or \code{write} seems to take effect instantaneously at some point between the call and the return. That is, each concurrent execution of multiple invocations corresponds to a \emph{linearization}--a valid sequence of those invocations where every \code{scan} returns the memory state resulting from all preceding \code{write} operations. 

The aforementioned sequential semantics of \code{scan} is defined by the action \code{scan\_spec} in \autoref{fig:snapshot}, which assumes that the return value equals some instantaneous read of the memory (recall that actions execute in a single indivisible step). 
Linearizability can be reduced to showing that \code{scan} is a refinement of \code{scan\_spec}, in a sense that will be made precise later. Note that \code{scan\_spec} may block and this is sound because linearizability does not imply any notion of progress by itself. 

\smallskip
\noindent
\textbf{Linearizability Proof.} 
The work of Attiya et al.\cite{DBLP:conf/wdag/AttiyaE19} shows that any ``unreduced'' linearizability proof requires prophecy variables, where an ``unreduced'' proof is one that attempts to establish a linearization for every possible execution of any number of invocations. 
Linearizability is equivalent to a standard notion of \emph{trace inclusion} between the concurrent object and an atomic (sequential) specification where every invocation performs a single indivisible step. Traces are sequences of call and return events storing input and return values. Trace inclusion is known to be equivalent to the existence of a composition of a forward and backward simulation relations~\cite{DBLP:journals/iandc/LynchV95}. Attiya et al.~\cite{DBLP:conf/wdag/AttiyaE19} show that there exists no forward simulation from this snapshot object to the corresponding atomic (sequential) specification, which implies the need for using backward simulations (a forward simulation corresponds to a proof using so-called fixed linearization points). 
Backward simulations are known to correspond to using prophecy variables in a deductive verification context~\cite{DBLP:journals/tcs/AbadiL91,DBLP:journals/iandc/LynchV95}.

Next, we present a proof using our reduction technique which avoids the use of notoriously challenging prophecy variables. After a step of abstraction, the reads will become movers and they can be reordered to form an atomic section which is a ``direct'' refinement of \code{scan\_spec}. 

This goes beyond previous reduction techniques since the \code{scan} implementation nests parallel composition and sequential composition of statements.

\subsection{Using Abstraction to Enable Reduction}

The actions \code{read} and \code{write} do not commute for obvious reasons. To enable reduction, we introduce two abstractions of \code{read}: a right-mover abstraction \code{read\_f} and a left-mover abstraction \code{read\_s}, which are listed on the right of \autoref{fig:snapshot:abs}. In general, an abstraction of an action over-approximates its effect on the global state and the set of possible return values.

A \code{read} abstraction commutes to the right of a write if it can return a value with an older timestamp than the one in memory, meaning any value it returns before the write remains valid afterward. We introduce this behavior via a non-deterministic choice: \code{read\_f} can either return the timestamped value in memory, or an arbitrary timestamped value provided that the timestamp is strictly smaller than the timestamp in memory. The left mover abstraction  \code{read\_s} is very similar except that the returned timestamp in the ``arbitrary'' case should be strictly higher than the timestamp in memory. We note that designing such abstractions requires creativity, as in any other deductive proof system. The advantage here is that they lead to more ergonomic proofs—both more succinct and less tedious. Moreover, soundness of mover abstractions is a local property, as it concerns only pairs of actions, and the induced reduction significantly simplifies subsequent reasoning. 

\autoref{fig:snapshot:abs} lists an abstraction of the \code{scan} procedure which calls the abstract actions \code{read\_f} and \code{read\_s} during the first and second read of the memory, resp. The occurrences of \code{seq-reduce} and \code{par-reduce} are explained below and should be ignored for now.
This is a sound abstraction in the sense that any concurrent execution of the original snapshot implementation is also possible when \code{scan} is replaced by this abstract version. Soundness is a straightforward consequence of the fact that \code{read\_f} and \code{read\_s} over-approximate the behavior of \code{read}.

In the following, we show that this abstraction of \code{scan} is a refinement of the atomic action \code{scan\_spec}, which concludes the linearizability proof.

\subsection{Reducing Parallel Statements}

We prove that it is sound to treat all reads in the abstract \code{scan} from \autoref{fig:snapshot:abs} as executing atomically, without interference from other threads. The first reduction step removes the use of parallel composition, represented by the annotation \code{par-reduce}. The first occurrence of \code{par-reduce} relies on \code{read\_f(2)} being a right mover and thus, 
\vspace{-1mm}
\begingroup
\small
\setlength{\jot}{2pt}
\begin{align*}
\code{(call r1[1] := read\_f(1)) par (call r1[2] := read\_f(2))} \\[-6mm]
\end{align*}
\endgroup
can be rewritten to 
\vspace{-2mm}
\begingroup
\small
\setlength{\jot}{2pt}
\begin{align*}
\code{call r1[1] := read\_f(1); call r1[2] := read\_f(2)} \\[-6mm]
\end{align*}
\endgroup

where the parallel composition \code{par} has been replaced by sequential composition \code{;}. \update{This fixes an order between the two actions, but interference is still allowed in between the two calls (i.e., after \code{read\_f(1)} completes but before \code{read\_f(2)} starts).}

Indeed, for any interleaving where \code{read\_f(2)} executes before \code{read\_f(1)}, the right moverness of \code{read\_f(2)} implies that it can be soundly swapped to the right of all actions that execute until \code{read\_f(1)} and \code{read\_f(1)} itself (here, soundness means preserving the final state and all return values of actions or procedures).

Dually, the second occurrence of \code{par-reduce} relies on \code{read\_s(1)} being a left mover in order to reduce the parallel composition \code{par} to sequential composition.

\begin{wrapfigure}{l}{0.32\textwidth}
\vspace{-7mm}
\begin{lstlisting}[numbers=none]
seq-reduce { // -> atomic {
  call r1[1] := read_f(1); // right
  call r1[2] := read_f(2); // right
  call r2[1] := read_s(1); // left
  call r2[2] := read_s(2); // left
  if (r1 == r2) {
    snapshot := r1;
    return;
  }
}
\end{lstlisting}
\vspace{-4mm}
\caption{A reduced loop iteration.}
\label{fig:snapshot:reduced}
\vspace{-7mm}
\end{wrapfigure}
The result of reducing the two parallel statements is shown on the left, in \autoref{fig:snapshot:reduced}. For simplicity, we write just the loop iteration. The sequence of reads is now a sequence of right movers followed by a sequence of left movers and we can use Lipton's reduction in order to rewrite it as an atomic section (the conditional and the assignments that follow the reads are accessing local variables and can be reordered in any direction, to the left in this case). Invoking this reduction principle is done via the keyword \code{seq-reduce}. 
The final reduced form of \code{scan} will group all reads and the if conditional inside an atomic section, marked using the keyword \code{atomic}. It is now quite straightforward to show that \code{scan} refines the atomic action \code{scan\_spec}: 
\vspace{-2mm}
\begin{itemize}
\item every iteration where the conditional fails has no effect on the return value,
\item if the conditional succeeds, then for every cell, the timestamps returned by \code{read\_f} and \code{read\_s} are identical. This indicates that both reads accessed the current memory state, and the values were not chosen arbitrarily. Specifically, if \code{read\_f} had returned a timestamp smaller than the one in memory, then \code{read\_s} could not have returned the same timestamp, as it only returns timestamps strictly greater than those currently in memory (which increase monotonically). The action \code{read\_s} cannot return a timestamp greater than the one in memory for similar reasons.
\end{itemize}
\vspace{-6mm}
\subsection{The Unbounded Memory Case: Reductions for Structured~Code}\label{sec:overview:unbounded}
\vspace{-2mm}
\begin{figure}[t]
\centering

\begin{minipage}{0.53\textwidth}
\begin{lstlisting}
procedure scan() returns (snapshot: [int]StampVal){
  var r1: [int]StampVal;
  var r2: [int]StampVal;
  while (true) {
    seq-reduce {
      call r1 := collect_f(n); // right
      call r2 := collect_s(n); // left
      if (r1 == r2) {
        snapshot := r1;
        return;
      }
    }
  }
}
\end{lstlisting}
\end{minipage}\hfill
\begin{minipage}{0.46\textwidth}
\begin{lstlisting}
right procedure collect_f(n: int) returns (r: [int]StampVal) {
  var out: StampVal;
  if (n == 0) { return; }
  else {
    par-reduce {
      (call r := collect_f(n-1)) par
      (call out := read_f(n))
    }
    r[n] := out;}
}

left procedure collect_s(n: int) returns (r: [int]StampVal) {
  var out: StampVal;
  if (n == 0) { return; }
  else {
    par-reduce {
      (call out := read_s(n)) par
      (call r := collect_s(n-1))
    }
    r[n] := out;}
}

\end{lstlisting}
\end{minipage}
\vspace{-2mm}
\caption{Applying reduction on an abstraction of the \code{scan} procedure for an unbounded size memory.}
\label{fig:snapshot:abs:unbounded}
\vspace{-7mm}
\end{figure}





\autoref{fig:snapshot:abs:unbounded} lists a reduction proof for an extension of the previous \code{scan} implementation to an unbounded size memory. Memory reads are performed inside two recursive procedures \code{collect\_f} and \code{collect\_s} which use the corresponding \code{read\_f} and \code{read\_s} actions to read memory cells. 

Notably, this demonstrates an extension of Lipton's reduction to structured programs, code that contains procedure calls. For compositionality, we introduce a moverness type for procedures, and use that moverness type in a similar way to Lipton's reduction. After a parallel reduction step (explained below), \code{collect\_f} and \code{collect\_s} are typed as right and left procedures, respectively. This enables a reduction step that yields an atomic section encompassing an entire iteration of the outer \code{scan} loop, as in the previous case.

The parallel reductions are now performed inside each of the two recursive procedures, and rely on similar arguments as above (the moverness of the \code{read\_f} and \code{read\_s} actions). After this reduction step, they contain no more parallel composition, and since all the actions they perform have the same moverness type, this type can be exported at the procedure level. The left moverness case requires that the procedure terminates \emph{when executed in isolation}, which is obvious here because the parameter decreases and it is bounded below by 0.

Once both reduction steps have been applied, proving that the \code{scan} procedure is a refinement of \code{scan\_spec} is similar to the bounded case presented above.

Formalizing the correctness of this reduction technique, which handles both structured code and parallel composition, is non-trivial. The next section introduces a simple yet expressive programming language used to reason about correctness in the following sections.
\vspace{-2mm}

\section{\lang: Syntax and Semantics}
\label{sec:language}
\vspace{-2mm}

\begin{figure}[t]
\centering
{\small 
$
\action \in \ActionName \quad
\proc \in \ProcName \quad
\procOrAction \in \ActionName \cup \ProcName
$

\begin{minipage}{.45\linewidth}
\centering
$
\begin{array}{rclcl}
               &   & \Val         &\ni& \undefbio\\
  \var         &\in& \Var         &=& \GlobalVar \cup \LocalVar \\
  \globalStore &\in& \GlobalStore &=& \GlobalVar \to \Val \\
  \localStore  &\in& \LocalStore  &=& \LocalVar \pto \Val \\
  \store       &\in& \Store       &=& \Var \pto \Val \\
  \gate        &\in& \Gate        &=& 2^\Store \\
  \trans       &\in& \Trans       &=& 2^{\Store \times \Store} \\
  \imap,\omap  &\in& \IOMap       &=& \LocalVar \pto \LocalVar \\
  \lvar        &\in& \LocalVar   & & \\
  \InVar,\OutVar &\in& 2^\LocalVar \\
  \MoverVar &\in& \MoverType &=& \{B, R, L, N, \top\} \\
  \mayFailVar &\in& \{\true, \false\}
\end{array}
$
\end{minipage}%
\begin{minipage}{.65\linewidth}
\centering
$
\begin{array}{rcl}
  \stmt \in \Stmt &::=& \skipstmt ~|~ \ifstmt{\lvar}{\stmt}{\stmt} \\
                  &   &|~ \seqstmt{\stmt}{\stmt} ~|~ \parstmt{\stmt}{\stmt} \\
                  &   &|~ \callstmt{(\procOrAction,\imap,\omap)} \\
                  &   &|~ \atomicstmt{\stmt} \\
                  &   &|~ \parreduce{\stmt}{\stmt} \\
                  &   &|~ \seqreduce{\stmt} \\
  \Action &::=& (\InVar, \OutVar, \gate, \trans, \MoverVar, \mayFailVar) \\
  \ProcSignature  &::=& (\InVar, \OutVar, \MoverVar, \mayFailVar) \\
  \as     &\in& \ActionName \to \Action \\
  \ps     &\in& \ProcName \to \ProcSignature \\
  \prog \in \Prog &=& \ProcName \to \Stmt
\end{array}
$
\end{minipage}
}
\vspace{-2mm}
\caption{\lang: Syntax}
\label{fig:syntax}
\vspace{-6mm}
\end{figure}


In this section we present our core programming language \lang to formalize our approach to reduction.
Our language is inspired by RefPL~\cite{DBLP:conf/cav/KraglQH20}.
\autoref{fig:syntax} summarizes the syntax of \lang.

\noindent
\textbf{Variables and stores.}
We assume there is a fixed set of \emph{global variables} $\GlobalVar$ and a fixed set of \emph{local variables} $\LocalVar$
such that $\GlobalVar$ and $\LocalVar$ are disjoint.
The set of \emph{variables} $\Var$ is the union of $\GlobalVar$ and $\LocalVar$.
A \emph{store} $\store$ is a partial map from variables to \emph{values}.
We write $\store' \subseteq \store$ if $\store$ is an extension of $\store'$,
$\store|_V$ for the restriction of $\store$ to $V$, $\store - V$ for $\store|_{dom(\store) \backslash V}$, 
$\store[\store']$ for the store that is like $\store'$ on $\dom(\store')$ and otherwise like $\store$,
and $\globalStore\cc\localStore$ for the combination of \emph{global store} $\globalStore$ and \emph{local store} $\localStore$.

\noindent
\textbf{Actions.}
\lang models uninterrupted execution by a thread using
atomic actions~\cite{DBLP:conf/popl/ElmasQT09,DBLP:conf/concur/KraglQH18}.
We assume there is a fixed set of actions with names from the set $\ActionName$.
\update{All accesses to global variables are confined to actions.}
The action map $\as$ maps each $\action \in \ActionName$ to a tuple $\as(\action) = (\InVar, \OutVar, \gate, \trans, \MoverVar, \mayFailVar)$.
The set of input variables $\InVar$ 
and the set of output variables $\OutVar$ are each a subset of $\LocalVar$.
The gate $\gate$ is a set of stores such that $\dom(\localStore) = \InVar$ for each $\globalStore\cc\localStore \in \gate$.
The transition relation $\trans$ is a relation over stores such that $\dom(\localStore) = \InVar$ and $\dom(\localStore') = \OutVar$
for each $(\globalStore\cc\localStore, \globalStore'\cc\localStore') \in \trans$.
Executing the action from a store $\store$ that does not satisfy the gate (i.e., $\store \notin \gate$) fails the execution. 
Otherwise, every transition $(\store,\store')$ in $\trans$ describes a possible atomic state transition from $\store$
(over $\GlobalVar \cup \InVar$) to $\store'$ (over $\GlobalVar \cup \OutVar$).
The mover type $\MoverVar$ of the action is a member of the set $\MoverType$~\cite{DBLP:conf/pldi/FlanaganQ03};
it captures succintly the nature of commutativity of this action compared to other actions defined by $\as$.
The failure type $\mayFailVar$ is a Boolean value that indicates whether it is possible for this action to fail.
If $\mayFailVar = \false$, then the gate must include all possible stores over $\GlobalVar \cup \InVar$.

Our formalization does not provide concrete syntax for the bodies of atomic actions,
instead choosing to model them abstractly using a symbolic transition system.
Our modeling approach is general and allows actions to be arbitrary and potentially failing computations
over global, input, and output variables.
Specifically, actions can model a variety of statements---asserts, assumes, (nondeterministic) assignments, choice, and sequencing.

\noindent
\textbf{Procedures.}
\lang models preemptible concurrent execution using procedures.
We assume there is a fixed set of procedures with names from the set $\ProcName$ which is disjoint from $\ActionName$.
We split the specification of procedures into two maps---the signature $\ps$ and the program $\prog$.
An important aspect of our formalization is to transform procedure bodies while keeping their signature fixed.
Splitting the specification of procedure behavior into the signature map and the program aids our formalization.

The signature map $\ps$ maps each $\proc \in \ProcName$ to a tuple $(\InVar, \OutVar, \MoverVar, \mayFailVar)$.
The set of input variables $\InVar$ 
and the set of output variables $\OutVar$ are each a subset of $\LocalVar$.
When $\proc$ is called, its local store gets a binding for each variable in $\LocalVar$.
The mover type $\MoverVar$ is a member of the set $\MoverType$.
The failure type $\mayFailVar$ is a Boolean indicating whether it is possible for the execution of the procedure to fail.
Our type checker, described later, checks the consistency of $\MoverVar$ and $\mayFailVar$ against the body of the procedure.

The program $\prog$ maps each $\proc \in \ProcName$ to a statement $\stmt$ that is executed when $\proc$ is called.
The primitive statement $\skipstmt$ does nothing; it serves as a marker in the formal operational semantics explained later.
The statement $\ifstmt{\lvar}{\stmt_1}{\stmt_2}$ looks up the value of $\lvar$ in the local store and continues to execute 
$\stmt_1$ if the value is $\true$ or $\stmt_2$ if the value is $\false$. 
The statement $\seqstmt{\stmt_1}{\stmt_2}$ executes $\stmt_1$ followed by $\stmt_2$.
The statement $\parstmt{\stmt_1}{\stmt_2}$ executes both $\stmt_1$ and $\stmt_2$ in parallel.

The statement $\callstmt{(\procOrAction,\imap,\omap)}$ calls either an action or a procedure.
Parameter passing is expressed using
an \emph{input map} $\imap$ from $\InVar$ to $\LocalVar$, and
an injective \emph{output map} $\omap$ from $\OutVar$ to $\LocalVar$.
For both $\imap$ and $\omap$, the domain is callee's formals and the range is caller's actuals.
Input variables are immutable, since they are not mapped to by output maps and the variables
of a procedure are not modified anywhere else.
An action call is the only way to access global variables or to modify either the global or the local store.

The statement $\atomicstmt{\stmt}$ executes $\stmt$ with preemptions disabled, i.e., the statement $\stmt$ is executed to
completion before any other concurrent execution is scheduled.  
The statement $\parreduce{\stmt_1}{\stmt_2}$ expresses the programmer intention to
reduce the parallel computation $\parstmt{\stmt_1}{\stmt_2}$ to the
sequential computation $\seqstmt{\stmt_1}{\stmt_2}$.
The statement $\seqreduce{\stmt}$ expresses the programmer intention to reduce the statement $\stmt$
to $\atomicstmt{\stmt}$.
A statement $s$ is {\em atomic-free} if $s$ does not have any occurrences of $\stmtfont{atomic}$.
A statement $s$ is {\em reduce-free} if $s$ does not have any occurrences of $\stmtfont{seq \mhyphen reduce}$
or $\stmtfont{par \mhyphen reduce}$.
A program $\prog$ is {\em atomic-free} if $\prog(Q)$ is atomic-free for all $Q \in \dom(\prog)$.

Although \lang does not have explicit support for loops and nondeterministic choice,
both can be modeled.
We can model a loop using a recursive procedure.
We can model nondeterministic choice using the conditional statement $\ifstmt{\lvar}{\stmt}{\stmt}$
after assigning a nondeterministically chosen value to the local variable $\lvar$ (via an action).

\begin{figure}
\centering
{\small

\newcommand{\rulename}[1]{\textbf{(#1)}}
\newcommand{\rulespace}{\vspace{4.3mm}}

\setlength{\abovedisplayskip}{0pt}
\setlength{\belowdisplayskip}{0pt}
\setlength{\abovedisplayshortskip}{0pt}
\setlength{\belowdisplayshortskip}{0pt}

\begin{minipage}{.5\linewidth}
\begin{align*}
  f           &::= (\localStore, \stmtCtxt[\stmt]) \\
  \Tree       &::= \Leaf{f} ~|~ \Node{f}{\ov{\Tree}}  \\
  \ThreadPool &::= \set{\Tree,\dots,\Tree} \\
  \conf       &::= (\globalStore, \ThreadPool) \mid \fail
\end{align*}
\end{minipage}%
\begin{minipage}{.5\linewidth}
\begin{align*}
  \stmtCtxt   &::= \hole_\stmt \mid \seqstmt{\stmtCtxt}{\stmt} \mid \inatomicstmt{\stmtCtxt} \\
              &\phantom{::==}\mid \inseqreduce{\stmtCtxt}\\
  \threadCtxt &::= \hole_\Tree \mid \Node{f}{\ov{\Tree} \threadCtxt \ov{\Tree}} \\
  \poolCtxt   &::= \set{\threadCtxt} \uplus \ThreadPool \\
  \leafCtxt   &::= \poolCtxt[\Leaf{(\hole_\localStore, \stmtCtxt)}]
\end{align*}
\end{minipage}

\rulespace

$
\begin{array}{ll}
  \rulename{proc call} &
  (\globalStore, \poolCtxt[\Leaf{(\localStore, \stmtCtxt[\callstmt[]{(\callee,\imap,\omap)}])}]) \step {} \\&
  (\globalStore, \poolCtxt[\Node{(\localStore, \stmtCtxt[\callstmt[]{(\callee,\imap,\omap)}])}{\Leaf{(\{\var\mapsto\undefbio \mid \var \in \LocalVar\}[\localStore\circ\imap], \prog(\callee))}}])
\end{array}
$

\rulespace

$
\begin{array}{ll}
  \rulename{return} &
  (\globalStore, \poolCtxt[\Node{(\localStore, \stmtCtxt[\callstmt[]{(\callee,\imap,\omap)}])}{\Leaf{(\callee,\localStore',\skipstmt)}}]) \step {} \\&
  (\globalStore, \poolCtxt[\Leaf{(\localStore[\localStore'\circ\inverse{\omap}], \stmtCtxt[\skipstmt])}])
\end{array}
$

\rulespace

$
\begin{array}{ll}
  \rulename{fork} &
  (\globalStore, \poolCtxt[\Leaf{(\localStore, \stmtCtxt[\parstmt{\stmt_1}{\stmt_2}])}]) \step {} \\&
  (\globalStore, \poolCtxt[\Node{(\localStore, \stmtCtxt[\parstmt{\stmt_1}{\stmt_2}])}{\Leaf{(\localStore,\stmt_1)}}\:{\Leaf{(\localStore,\stmt_2)}}])
\end{array}
$

\rulespace

$
\begin{array}{ll}
  \rulename{join} &
  (\globalStore, \poolCtxt[\Node{(\localStore, \stmtCtxt[\parstmt{\stmt_1}{\stmt_2}])}{\Leaf{(\localStore_1,\skipstmt)}}\:{\Leaf{(\localStore_2,\skipstmt)}}]) \step {} \\&
  (\globalStore, \poolCtxt[\Leaf{(\localStore[\localStore_1|_{\modified(\stmt_1)}][\localStore_2|_{\modified(\stmt_2)}], \stmtCtxt[\skipstmt])}])
\end{array}
$

\rulespace

$
\infer
{
  (\globalStore,  \leafCtxt[\localStore][\callstmt{(\action,\imap,\omap)}]) \step
  (\globalStore', \leafCtxt[\localStore'][\skipstmt])
}
{
  \begin{gathered}
  \rulename{action call}~
  as(A) = (\blank, \blank, \gate, \trans) \quad
  (\globalStore\cc(\localStore \circ \imap), \globalStore'\cc\hat{\localStore}) \in \gate\circ\trans \\
  \localStore' = \localStore[\hat{\localStore} \circ \inverse{\omap}]
  \end{gathered}
}
$

\rulespace

$
\infer
{
  (\globalStore, \leafCtxt[\localStore][\callstmt{(\action,\imap,\omap)}]) \step
  \fail
}
{
  \begin{gathered}
  \rulename{action fail}~
  \as(\action) = (\blank, \blank, \gate, \blank, \blank, \blank) \\
  \globalStore\cc(\localStore\circ\imap) \notin \gate
  \end{gathered}
}
$
\hfill
$
\infer
{
  (\globalStore, \leafCtxt[\localStore][\ifstmt{\lvar}{\stmt_1}{\stmt_2}]) \step
  (\globalStore, \leafCtxt[\localStore][\stmt'])
}
{
  \rulename{branch}\quad
  \stmt' =
  \begin{cases}
    \stmt_1 & \localStore[\lvar] = \true\\
    \stmt_2 & \localStore[\lvar] = \false
  \end{cases}
}
$

\rulespace

$
\rulename{skip}~
(\globalStore, \leafCtxt[\localStore][\seqstmt{\skipstmt}{\stmt}]) \step
(\globalStore, \leafCtxt[\localStore][\stmt])
$
\hfill
$
\rulename{stop}~
(\globalStore, \set{\Leaf{(\blank, \skipstmt)}} \uplus \ThreadPool) \step
(\globalStore, \ThreadPool)
$

\rulespace

$
\rulename{atomic enter}~
(\globalStore, \leafCtxt[\localStore][\atomicstmt{\stmt}]) \step
(\globalStore, \leafCtxt[\localStore][\inatomicstmt{\stmt}])
$

\rulespace
$
\rulename{atomic exit}~
(\globalStore, \leafCtxt[\localStore][\inatomicstmt{\skipstmt}]) \step
(\globalStore, \leafCtxt[\localStore][\skipstmt])
$

\rulespace

$
\begin{array}{ll}
  \rulename{par-reduce enter} &
  (\globalStore, \poolCtxt[\Leaf{(\localStore, \stmtCtxt[\parreduce{\stmt_1}{\stmt_2}])}]) \step {} \\&
  (\globalStore, \poolCtxt[\Node{(\localStore, \stmtCtxt[\parreduce{\stmt_1}{\stmt_2}])}{\Leaf{(\localStore,\stmt_1)}}\:{\Leaf{(\localStore,\stmt_2)}}])
\end{array}
$

\rulespace

$
\begin{array}{ll}
  \rulename{par-reduce exit} &
  (\globalStore, \poolCtxt[\Node{(\localStore, \stmtCtxt[\parreduce{\stmt_1}{\stmt_2}])}{\Leaf{(\localStore_1,\skipstmt)}}\:{\Leaf{(\localStore_2,\skipstmt)}}])\step {} \\&
  (\globalStore, \poolCtxt[\Leaf{(\localStore[\localStore_1|_{\modified(\stmt_1)}][\localStore_2|_{\modified(\stmt_2)}], \stmtCtxt[\skipstmt])}])
\end{array}
$

\rulespace

$
\rulename{seq-reduce enter}~
(\globalStore, \leafCtxt[\localStore][\seqreduce{\stmt}]) \step
(\globalStore, \leafCtxt[\localStore][\inseqreduce{\stmt}])
$

\rulespace
$
\rulename{seq-reduce exit}~
(\globalStore, \leafCtxt[\localStore][\inseqreduce{\skipstmt}]) \step
(\globalStore, \leafCtxt[\localStore][\skipstmt])
$

\caption{\lang: Operational semantics for program $\prog$}
\label{fig:semantics}

}
\end{figure}

\subsection{Semantics}

\autoref{fig:semantics} presents the operational semantics of \lang as 
a transition relation $\step$ over \emph{configurations}.
Each configuration is either a failure $\fail$ or a pair $(\globalStore, \ThreadPool)$
comprising a global store $\globalStore$ and a finite multiset $\ThreadPool$ of threads.
Each thread is a tree (which generalizes a call stack);
new leaf nodes (\textsf{Lf}) are created via the $\callstmtempty$ and $\parstmtempty$ statements. 
Both of these statements block the caller in an internal node \textsf{Nd} until the leaf nodes are finished.
Each tree node contains a \emph{frame} $(\localStore, \stmt)$, where $\localStore$ is the local store and $\stmt$ is the remaining statement to execute. 

In the definition of $\step$ we use several evaluation contexts that have a unique hole $\hole$ which marks the evaluation position;
filling the hole is denoted by $\cdot[\cdot]$. 
$\stmtCtxt$ is the statement context, which is a statement with a hole $\hole$.
$\stmtCtxt[\stmt]$ is a statement with $\stmt$ in evaluation position.
In addition to a hole $\hole$, there are three other statement contexts.
The context $\seqstmt{\stmtCtxt}{\stmt}$ finishes evaluating $\stmtCtxt$ before moving on to $\stmt$.
The context $\inatomicstmt{\stmtCtxt}$ is introduced when $\inatomicstmtempty$ is entered.
The context $\inseqreduce{\stmtCtxt}$ is introduced when $\inseqreduceempty$ is entered.
$\poolCtxt$ is a multiset of thread trees with a hole $\hole$ in one of the trees.
We have additional conditions on this multiset $\poolCtxt$:
(1)~Trees that do not contain a hole $\hole$ do not have the $\inatomicstmtempty$ statement in them. 
(2)~If $\Tree$ is the tree with the hole $\hole$, and if there is any $\inatomicstmtempty$ statement in $\Tree$,
then it must be on the unique path from root of $\Tree$ to the hole. 
These conditions ensure that an atomic statement executes without interference.
The hole in a $\poolCtxt$ may be filled with an arbitrary tree.
$\leafCtxt$ is a specialization of $\poolCtxt$ in which the hole is filled with a leaf with holes inside it
for a local store and the next statement to be executed. 

\begin{wrapfigure}{l}{0.57\textwidth}
\vspace{-12mm}
\centering
\setlength{\arraycolsep}{2pt}
{\small
\[
\begin{array}{@{}rcl@{}}
\modified(\skipstmt) & = & \emptyset \\
\modified(\callstmt(X,\imap,\omap)) & = & \img(\omap) \\
\modified(\seqstmt{\stmt_1}{\stmt_2}) & = & \modified(\stmt_1) \cup \modified(\stmt_2) \\
\modified(\parstmt{\stmt_1}{\stmt_2}) & = & \modified(\stmt_1) \cup \modified(\stmt_2) \\
\modified(\atomicstmt{\stmt}) & = & \modified(\stmt) \\
\modified(\parreduce{\stmt_1}{\stmt_2}) & = & \modified(\stmt_1) \cup \modified(\stmt_2) \\
\modified(\seqreduce{\stmt}) & = & \modified(\stmt) \\
\modified(\ifstmt{\lvar}{\stmt_1}{\stmt_2}) & = & \modified(\stmt_1) \cup \modified(\stmt_2) \\
\end{array}
\]
}
\vspace{-5mm}
\caption{The $\modified$ function}
\label{fig:modified}
\vspace{-8mm}
\end{wrapfigure}

\autoref{fig:semantics} provides the semantics for a fixed program $\prog$ organized as a collection of rules,
one for each kind of statement in the hole.
The rule $\textbf{(proc call)}$ for executing $\callstmt[]{(\callee,\imap,\omap)}$ from a context with local store $\localStore$
creates a new leaf and initializes its frame with a local store where the input variables of $\callee$
get their values from $\localStore \circ \imap$ ($\circ$ means function or relation composition)
and the rest of the variables are initialized to the default value $\undefbio$.
The rule $\textbf{(return)}$ for returning from a call updates the caller's local store with the values in the callee's
local store using the output map $\omap$.

The rule $\textbf{(fork)}$ for executing $\parstmt{\stmt_1}{\stmt_2}$ creates two leaf nodes for executing $\stmt_1$ and $\stmt_2$,
each node getting a copy of the local store of the parent.
The parent is blocked until both children nodes have finished executing.
Then, the modified variables from each child node are written back to the parent's local store in the rule $\textbf{(join)}$.
The type checker for \lang (described in \autoref{sec:types}) checks that the local variables modified
by $\stmt_1$ are not accessed by $\stmt_2$ and vice-versa.
This check ensures that the updates to the local store of the parent from the local stores of the children nodes
are conflict-free.
The $\modified$ function, shown in \autoref{fig:modified}, approximates the set of local variables that are modified by a statement.

Atomic actions execute directly in the context of the caller.
If the current store does not satisfy the gate of an executed action,
the execution stops in the \emph{failure configuration}~$\fail$ (rule $\textbf{(action fail)}$).
Otherwise, the execution takes a step according to the transition relation of the action (rule $\textbf{(action call)}$).
We refer to transitions in which an atomic action executes as an $action$ transition,
and all other transitions as $local$ transitions.
For any action transition $t_A$ of an action $A$ called with input mapping $\imap_A$ from a leaf node with local store $\ell$ in the frame,
we define the input parameter binding of the transition as $\inputbinding(t_A) = \ell \circ \imap_A$.

Rule $\textbf{(branch)}$ executes a conditional statement.
Rule $\textbf{(skip)}$ moves the evaluation context forward.
Rule $\textbf{(stop)}$ removes a finished tree from the multiset of trees.
Rule $\textbf{(atomic enter)}$ enters an atomic section and rule $\textbf{(atomic exit)}$ exits it.
Rules $\textbf{(par-reduce enter)}$ and $\textbf{(par-reduce exit)}$ are similar to $\textbf{(fork)}$ and $\textbf{(join)}$, respectively.
Rules $\textbf{(seq-reduce enter)}$ and $\textbf{(seq-reduce exit)}$ enter and exit a $\seqreduceempty$ block.


\section{Commutativity of Atomic Actions}
\label{sec:movers}
\vspace{-2mm}
We present basic concepts that will be used towards the end of this section 
to define a well-formed action map.
Intuitively, an action map $\as$ is well-formed if the mover type of each
atomic action is correct w.r.t. the entire pool of atomic actions in $\as$.
For the next few definitions, we fix two actions
$X = (\InVar_X, \OutVar_X, \gate_X, \trans_X, \_, \_)$ and $Y = (\InVar_Y, \OutVar_Y, \gate_Y, \trans_Y, \_, \_)$,

\vspace{-2mm}
\paragraph*{Weakest liberal precondition}
The weakest liberal precondition $\wlp(X, \gate_Y)$ is the set of all triples $(\globalStore, \ell_X, \ell_Y)$ such that
$X$ does not fail in $\globalStore \cc \ell_X$ 
and executing $X$ from $\globalStore \cc \ell_X$ always leads to a global store $\globalStore'$ where gate of Y holds
on $\globalStore' \cc \ell_Y$.
\begingroup
\small
\setlength{\jot}{2pt}
\begin{align*}
\wlp(X, \gate_Y)
&= \{ (\globalStore, \ell_X, \ell_Y) \mid
      \globalStore \cc \ell_X \in \gate_X \;\land \\
&\qquad
   \bigl(
   \forall \globalStore', \ell_X' :
   (\globalStore \cc \ell_X, \globalStore' \cc \ell_X') \in \trans_X
   \;\implies\;
   \globalStore' \cc \ell_Y \in \gate_Y
   \bigr)
   \}
\end{align*}
\endgroup

A consequence of this definition is that if a state satisfies the gate of $X$ but not $\wlp(X, \gate_Y)$,
then there is a way to execute $X$ and get to a state where gate of $Y$ does not hold.
This consequence, noted formally below, is used heavily in the proof of soundness of our reduction theorem.
\begingroup
\small
\setlength{\jot}{2pt}
\begin{align*}
\forall g, \ell_X, \ell_Y:\;&
   g \cc \ell_X \in \gate_X \;\land\;
   (g, \ell_X, \ell_Y) \notin \wlp(X, \gate_Y)
   \;\implies \\
&\qquad
   \exists \hat{g}, \hat{\ell}:\;
   (g \cc \ell_X, \hat{g} \cc \hat{\ell}) \in \trans_X
   \;\land\;
   \hat{g} \cc \ell_Y \notin \gate_Y
\end{align*}
\endgroup

\paragraph*{Commutes}
We say $X$ commutes with $Y$, denoted by $\commutes(X, Y)$,
if executing $X$ followed by $Y$ leads to a state that is also possible by executing $Y$ before $X$.
\begingroup
\small
\setlength{\jot}{2pt}
\begin{align*}
\forall g,g',\bar g,\ell_X,\ell_Y,\ell_X',\ell_Y' :\;&
(g,\ell_X,\ell_Y)\in\wlp(X,\gate_Y)\land(g,\ell_Y,\ell_X)\in\wlp(Y,\gate_X)\\
&\land (g\cc\ell_X,\bar g\cc\ell_X')\in(\gate_X\circ\trans_X)\land(\bar g\cc\ell_Y,g'\cc\ell_Y')\in(\gate_Y\circ\trans_Y)\\
&\implies \exists \hat g:\;
(g\cc\ell_Y,\hat g\cc\ell_Y')\in(\gate_Y\circ\trans_Y)\land(\hat g\cc\ell_X,g'\cc\ell_X')\in(\gate_X\circ\trans_X)
\end{align*}
\endgroup

\paragraph*{Success preservation}
We say $X$ \emph{preserves the success} of $Y$, denoted by \\ $\preservesSuccess(X, Y)$,
if whenever gate of $Y$ and gate of $X$ hold from a state, then any transition of $X$ leads to a state that also satisfies gate of $Y$.
\vspace{-1mm}
\begingroup
\small
\setlength{\jot}{2pt}
\begin{align*}
    \forall  g, \ell_X, \ell_Y : \;
    g\cc\ell_X \in \gate_X \land g\cc\ell_Y \in \gate_Y \implies (g, \ell_X, \ell_Y) \in \wlp(X, \gate_Y) \\[-6mm]
\end{align*}
\endgroup

\paragraph*{Failure preservation}
We say $X$ \emph{preserves the failure} of $Y$, denoted by\\ $\preservesFailure(X, Y)$,
if whenever $X$ does not fail from a state but $Y$ does, there exists a transition of $X$ from that state that leads to the failure of $Y$.
Equivalently, if $X$ does not fail from a state and cannot make a transition that leads to the failure of $Y$, then $Y$ does not fail
from that state. 
\begingroup
\small
\setlength{\jot}{2pt}
\begin{align*}
    \forall  g, \ell_X, \ell_Y : \;
    (g, \ell_X, \ell_Y) \in \wlp(X, \gate_Y) \implies g\cc\ell_Y \in \gate_Y
\end{align*}
\endgroup
\noindent
\textbf{Well-formed action map.}
We have 5 types: right-mover ($\rightmover$), left-mover ($\leftmover$), both-mover ($\bothmover$), non-mover ($\nonmover$), top ($\top$). 
These types form a lattice under a partial order defined as: $\bothmover \sqsubseteq a \sqsubseteq \nonmover \sqsubseteq \top$
where $a \in \{\rightmover, \leftmover\}$.
Let the join operator for this partial order be denoted by  $\sqcup$.
An action map $as$ is \emph{well-formed} if for all action $A \in \ActionName$,
if $\as(A)=(\InVar_A, \OutVar_A, \gate_A,\trans_A, \MoverVar_A, \mayFailVar_A)$,
then the following conditions are satisfied:




\vspace{-3mm}
{\small
\setlength{\abovedisplayskip}{4pt}
\setlength{\belowdisplayskip}{4pt}
\setlength{\abovedisplayshortskip}{2pt}
\setlength{\belowdisplayshortskip}{2pt}
\setlength{\itemsep}{2pt}
\setlength{\topsep}{2pt}
\begin{enumerate}
\item $M_A \neq \top$.

\item If $M_A \sqsubseteq \leftmover$, then for all $X \in \ActionName$:
\[
\begin{array}{@{}l@{\qquad}l@{}}
\condition{L1} & \preservesSuccess(X, A) \\
\condition{L2} & \preservesFailure(A, X) \\
\condition{L3} & \commutes(X, A)
\end{array}
\]

\item If $M_A \sqsubseteq \rightmover$, then for all $X \in \ActionName$:
\[
\begin{array}{@{}l@{\qquad}l@{}}
\condition{R1} & \preservesSuccess(A, X) \\
\condition{R2} & \commutes(A, X)
\end{array}
\]

\item If $\mayFailVar_A = \false$, then
$
\gate_A
=
\{\, \globalStore \cc \ell_A \mid
\globalStore \in \GlobalStore \land
\ell_A \in \LocalStore \land
\dom(\ell_A) = \InVar_A \,\}.
$
\end{enumerate}
}
\vspace{-1mm}

There are important differences between our conditions for a well-formed action map
and those in prior work on Civl~\cite{DBLP:conf/cav/KraglQ18}. Our conditions are weaker and we were able to identify these in the process of writing the proof for the soundness of our reduction techniques (\autoref{thm:maintheorem}).

First, our definition of $\commutes(X, Y)$ is weaker.
The corresponding condition for Civl only assumes $\gate_X$ and $\gate_Y$ in the initial state
while verifying that $X$ and $Y$ can be reordered.
In contrast, our condition makes the stronger assumption $\wlp(X, \gate_Y)$ and $\wlp(Y, \gate_X)$ in the initial state.
Intuitively, our check requires reordering $X$ and $Y$ only for those initial states
from which it is impossible to fail for any ordering of $X$ and $Y$.
\edit{The following example illustrates this difference. Let $S$ be a shared set, and define actions $A$ and $B$ by}
$
A(j):\; S := S \setminus \{j\},\;
B(i):\; \texttt{assert }\, i \in S
$
\edit{where the assertion is the gate of $B$. Under our commutativity condition, $A$ is a right mover and $B$ is a left mover, but this does not hold under Civl's original condition.}

Second, Civl performs two separate checks, backward preservation and nonblocking, for left movers.
Backward preservation requires a left mover $A$ to demonstrate for every action $X$ that
if $\gate_X$ holds after a step of $A$ then $\gate_X$ also holds before the step.
The nonblocking check requires $A$ to either fail or take a step from every initial state.
Instead, we have a single failure preservation check that requires less of $A$ than the combination of
backward preservation and nonblocking requirements.
Intuitively, our check requires $A$ to take a step only when trying to preserve a failure of the action $X$.
For example, if $X$ does not have any failing behavior, failure preservation would hold trivially,
but nonblocking check for $A$ could still be nontrivial.
\vspace{-2mm}

\section{Types for Reduction}
\label{sec:types}
\vspace{-2mm}
In this section, we exploit mover types of atomic actions to check that
the application of $\parreduceempty$ and $\seqreduceempty$ in a program
are applicable to achieve sound reduction.
We achieve this goal by defining two helper functions on statements---$\mayfail$ and $\movertype$.
The function $\mayfail$ propagates the failure types of actions to statements by using
a conservative static analysis.
The function $\movertype$ lifts mover types of actions to statements using an
effect system~\cite{DBLP:conf/pldi/FlanaganQ03}.
Together, these functions allow us to define a well-typed program which implies that
applications of reduction in the program are valid.
\vspace{-3mm}
\paragraph{\bf{Failure typing for statements}}
We compute $\mayfail(\stmt)$ for any statement $\stmt$ by conservatively propagating the failure types of atomic actions.
\begin{figure}
\vspace{-10mm}
\centering
\setlength{\arraycolsep}{2pt}
{\small
\[
\begin{array}{rcl}
\mayfail(\skipstmt) &=& \false \\
\mayfail(\callstmt(\action,\imap,\omap)) &=& \mayfail(\action) \\
\mayfail(\callstmt(\proc,\imap,\omap)) &=& \mayfail(\proc) \\
\mayfail(\seqstmt{\stmt_1}{\stmt_2}) &=& \mayfail(\stmt_1) \lor \mayfail(\stmt_2) \\
\mayfail(\parstmt{\stmt_1}{\stmt_2}) &=& \mayfail(\stmt_1) \lor \mayfail(\stmt_2) \\
\mayfail(\atomicstmt{\stmt}) &=& \mayfail(\stmt) \\
\mayfail(\parreduce{\stmt_1}{\stmt_2}) &=& \mayfail(\stmt_1) \lor \mayfail(\stmt_2) \\
\mayfail(\seqreduce{\stmt}) &=& \mayfail(\stmt) \\
\mayfail(\ifstmt{\lvar}{\stmt_1}{\stmt_2}) &=& \mayfail(\stmt_1) \lor \mayfail(\stmt_2) \\
\end{array}
\]
}
\vspace{-5mm}
\caption{The $\mayfail$ function.}
\label{fig:mayfail}
\vspace{-8mm}
\end{figure}

\vspace{-3mm}
\paragraph{\bf{Mover typing for statements}}
Given a well-formed action map $\as$ and a procedure signature map $\ps$, we define $\movertype$ function,
which assigns a mover type to a statement.
To define this function, we first define the sequential composition of mover types in the table below.
Using this table, we can define $\movertype$ recursively as follows:

\vspace{-7mm}
\begin{figure}
\vspace{-2mm}
\begin{minipage}{.6\textwidth}
{\small
\[
\begin{array}{rcl}
\movertype(\skipstmt) &=& \bothmover \\
\movertype(\parstmt{\stmt_1}{\stmt_2}) &=& \top \\
\movertype(\callstmt(\action,\imap,\omap)) &=& \movertype(\action) \\
\movertype(\callstmt(\proc,\imap,\omap)) &=& \movertype(\proc) \\
\movertype(\atomicstmt{\stmt}) &=& \movertype(\stmt) \\
\movertype(\seqreduce{\stmt}) &=& \movertype(\stmt) \\
\movertype(\parreduce{\stmt_1}{\stmt_2}) &=& \movertype(\stmt_1;\stmt_2) \\
\movertype(\seqstmt{\stmt_1}{\stmt_2}) &=& \movertype(\stmt_1);\movertype(\stmt_2) \\
\movertype(\ifstmt{\lvar}{\stmt_1}{\stmt_2}) &=& \movertype(\stmt_1) \sqcup \movertype(\stmt_2) \\
\end{array}
\]
}
\end{minipage}
\begin{minipage}{.35\textwidth}
\vspace{-9mm}
{\small
\begin{center}
    \begin{tabular}{ | c | c c c c  c |} 
     \hline
     ; & B & L & R & N & $\top$ \\ 
     \hline
     B & B & L & R & N & $\top$ \\ 
     R & R & N & R & N & $\top$ \\
     L & L & L & $\top$ & $\top$ & $\top$ \\
     N & N & N & $\top$ & $\top$ & $\top$ \\
     $\top$ & $\top$ & $\top$ & $\top$ & $\top$ & $\top$ \\
     \hline
    \end{tabular}
\end{center}
}
\end{minipage}
\vspace{-2mm}
\caption{The $\movertype$ function.}
\vspace{-6mm}
\end{figure}

\smallskip
For example, consider the following seq-reduce statement from the example in the overview.
The mover type assigned to the seq-reduce block will be $\nonmover$.

\lstset{frame=none, numbers=none}
\begin{lstlisting}
seq-reduce {
    par-reduce {(call r1[1] := read_f(1)) par (call r1[2] := read_f(2))};
    par-reduce {(call r2[1] := read_s(1)) par (call r2[2] := read_s(2))};
}
\end{lstlisting}

Each call to {\tt read\_f} within the first par-reduce has $\rightmover$ type.
The par-reduce statement then gets the mover type of $\rightmover;\rightmover$ = $\rightmover$ from the table.
Similarly, the second par-reduce contains calls to {\tt read\_s}, which are each $\leftmover$,
and the par-reduce statement gets the mover type of $\leftmover;\leftmover$ = $\leftmover$.
Now, the statement inside seq-reduce composes the two par-reduce blocks sequentially: the first has type $\rightmover$, and the second has type $\leftmover$.
The sequential composition $\rightmover;\leftmover$ results in an $\nonmover$ type, which is then assigned to the entire seq-reduce block.

The rules above imply that if $\stmt$ is nested inside $\stmt'$ and $\movertype(\stmt) = \top$ then $\movertype(\stmt') = \top$.
Since $\movertype(\parstmt{\stmt_1}{\stmt_2}) = \top$, 
if $\parstmt{\stmt_1}{\stmt_2}$ is nested inside a statement $\stmt$, then $\movertype(\stmt) = \top$.
We will return to this observation when we discuss the rules for well-typed programs below.
\vspace{-3mm}
\paragraph{\bf{Well-typed programs}}
We define a predicate $\welltyped(\stmt, \localvarset)$ where $\stmt \in \Stmt$ and $\localvarset \in 2^{\LocalVar}$.
The predicate $\welltyped(\stmt, \localvarset)$ is checking two aspects of a statement.
First, it checks that $\stmt$ only accesses the local variables it is allowed to.
This check is relevant because if $\stmt$ contains a nested occurrence of two statements $\stmt_1$ and $\stmt_2$ executing in parallel,
then local variables modified by these two statements must be disjoint from each other.
In fact, we check a stronger, yet simpler to check, property that local variables modified by $\stmt_1$ are neither read nor written by $\stmt_2$,
and vice-versa.
\update{This requirement is used only to simplify the formalization. One could consider fresh copies of local variables instead (note that the number of arms in a parallel construct is constant). In our implementation, the arms of the par construct can only be procedure calls, so this restriction just translates to requiring their input and output parameters to be disjoint, which is both natural and easy to check.}
Second, we check that applications of $\parreduceempty$ and $\seqreduceempty$ have the appropriate mover and failure types on their target statements.
This part of the check uses the previously defined functions $\movertype$ and $\mayfail$.
\vspace{-1mm}
\begin{figure}
\vspace{-8mm}
\centering
\setlength{\arraycolsep}{2pt}
{\small
\[
\begin{array}{rcl}
\welltyped(\skipstmt, \localvarset) &=& true \\
\welltyped(\atomicstmt{\stmt}, \localvarset) &=& \welltyped(\stmt, \localvarset) \\
\welltyped(\callstmt{(\procOrAction, \imap, \omap)}, \localvarset) &=& \img(\imap) \subseteq \localvarset \land \img(\omap) \subseteq \localvarset \\
\welltyped(\ifstmt{\lvar}{\stmt_1}{\stmt_2}, \localvarset) &=& \welltyped(\stmt_1, \localvarset) \land \welltyped(\stmt_2, \localvarset) \land \lvar \in \localvarset \\
\welltyped(\seqstmt{\stmt_1}{\stmt_2}, \localvarset) &=& \welltyped(\stmt_1, \localvarset) \land \welltyped(\stmt_2, \localvarset) \\
\welltyped(\parstmt{\stmt_1}{\stmt_2}, \localvarset) &=& \welltyped(\stmt_1, \localvarset - \modified(\stmt_2)) \land \\
                                                     & & \welltyped(\stmt_2, \localvarset - \modified(\stmt_1)) \\
\welltyped(\seqreduce{\stmt}, \localvarset) &=& \welltyped(\stmt, \localvarset) \land \movertype(\stmt) \sqsubseteq \nonmover \\
\welltyped(\parreduce{\stmt_1}{\stmt_2}, \localvarset) &=& \welltyped(\stmt_1, \localvarset - \modified(\stmt_2)) \land \\
                                                       & & \welltyped(\stmt_2, \localvarset - \modified(\stmt_1)) \land \\ 
                                                       & & (\movertype(\stmt_1) \sqsubseteq \leftmover \vee \\
                                                       & & \hspace{2mm}\movertype(\stmt_2) \sqsubseteq \rightmover \land \neg \mayfail(\stmt_2)) 
\end{array}
\]
}
\vspace{-4mm}
\caption{The $\welltyped$ function}
\label{fig:welltyped}
\vspace{-7mm}
\end{figure}

Most rules in the definition above are straightforward.
We take a closer look at the rules for parallel and sequential reduction, focusing on the mover-related checks.
$\parreduce{\stmt_1}{\stmt_2}$ checks that one of two cases apply:
either $\stmt_1$ is a left mover or $\stmt_2$ is a right mover and must not fail.
Intuitively, in both cases, all the code in $\stmt_1$ may be commuted before all the code in $\stmt_2$.
\update{For instance, here is an example illustrating why the right mover statement $\stmt_2$ is not allowed to fail.}
Let \code{x} be a shared integer variable, and let \code{read} and \code{inc} be actions, where \code{read} is a right-mover that first asserts \code{x > 0} (this corresponds to its gate) and then reads the value of \code{x} (exactly or less), and \code{inc} is a non-mover that increments \code{x} by 1 (it is not a left-mover because it does not satisfy failure preservation, and it is not a right-mover because it does not commute to the right of a read). 
Let Q be a procedure in the original program P, and assume that it gets reduced to Q' in the reduced program P' by an application of par-reduce (as shown in the snippet above). Assuming an initial state where \code{x = 0}, $P$ can fail (if \code{read} executes first) but $P'$ has no execution that fails (since \code{inc} always executes before \code{read}). This is problematic since the reduction ``hides'' failures which goes against sound reasoning.
Hence we require that right movers do not fail. This is what the helper function $\mayfail$ checks.

\begin{figure}[t]
\centering

\begin{minipage}[t]{0.52\textwidth}
\begin{lstlisting}
var x: int;

right action read() returns (out: int) {
  assert x > 0;
  assume out <= x;
}

action inc() {
  x := x + 1;
}
\end{lstlisting}
\end{minipage}\hfill
\begin{minipage}[t]{0.46\textwidth}
\begin{lstlisting}
procedure Q {
  par-reduce {
    call inc() par (call read())
  }
}

procedure Q' {
  call inc();
  call read();
}
\end{lstlisting}
\end{minipage}
\caption{Example illustrating need for right movers do not fail condition}
\vspace{-4mm}
\end{figure}


$\seqreduce{\stmt}$ checks that $\movertype(\stmt) \sqsubseteq \nonmover$ which implies that any execution
through $\stmt$ is of a form $\rightmover^* \cc \nonmover? \cc \leftmover^*$,
and therefore $\stmt$ can be converted to an atomic section~\cite{DBLP:conf/pldi/FlanaganQ03}.
In this case, the statement $\stmt$ must not have any unreduced parallel statement (whose mover type is $\top$) nested inside it.
But it is possible for $\stmt$ to have a reduced parallel statement (whose mover type could be different from $\top$) nested inside it.
This flexibility is important in practice and is particularly useful for our case studies described in \autoref{sec:evaluation}.

A program $\prog$ is {\em well-typed} if for all $Q \in \dom(\prog)$, if $\ps(Q) = (\_, \_, M, \mayFailVar)$ then
the following hold:
(1)~$\welltyped(\prog(Q), \LocalVar)$,
(2)~$\movertype(\prog(Q)) \sqsubseteq M$, and
(3)~$\mayfail(\prog(Q)) \Rightarrow \mayFailVar$.
The well-typed predicate is computed separately for each procedure $Q$ in the program using the signatures of all procedures
and actions called by $Q$.
First, the body of $Q$ is checked to be well-typed w.r.t. the set of all local variables.
Second, the mover type of the body of $Q$ must be stronger than the annotated mover type $M$ of $Q$.
This check ensures that the type checking of procedures that call $Q$ will succeed even if $Q$ was inlined at the call site.
In essence, this means that the mover type of any procedure is valid regardless of the specific execution taken in completing a
call to that procedure.
Third, the failure type of $Q$ is checked to be a conservative approximation of the failure type of the body of $Q$.
We use the notion of well-typed programs in \autoref{sec:refinement} to state the soundness theorem of our reduction technique.

\section{Reduction for \textsf{RedPL} Programs}
\label{sec:refinement}

In this section, we give a meaning to the $\stmtfont{seq \mhyphen reduce}$ and $\stmtfont{par \mhyphen reduce}$ annotations in \lang programs, and state the related soundness theorem. Soundness is stated in terms of a \emph{refinement} relation between programs that we define hereafter.

A configuration $(\globalStore, \ThreadPool)$ is {\em initial} if it contains an arbitrary number of threads that are about to execute a well-typed statement, i.e., for all $\Tree \in \ThreadPool$,
there exist $\localStore$ and $\stmt$ such that
$\Tree = Lf (\localStore, \stmt)$, $\welltyped(\stmt, \LocalVar)$, and $\stmt$ is atomic-free and reduce-free.
A configuration $(\globalStore, \ThreadPool)$ is {\em final} if $\ThreadPool = \emptyset$,
i.e., all threads have finished executing successfully.
The failure configuration $\fail$ is also final.

Given two well-typed programs $\prog$ and $\prog'$, we say $\prog$ refines $\prog'$ (denoted $\prog \refines \prog'$) if the following two properties hold for all initial configurations $(\globalStore, \ThreadPool)$:
\\
\condition{P1}~ If there is an execution of $\prog$ that fails from the initial configuration $(\globalStore, \ThreadPool)$, then there also is an execution of $\prog'$ that fails from the same initial configuration:
\vspace{-5mm}
\begingroup
\small
\setlength{\jot}{2pt}
\begin{align*}
        (\globalStore, \ThreadPool) \stepp[\prog]\transitive \fail
    \implies (\globalStore, \ThreadPool) \stepp[\prog']\transitive \fail
\end{align*}
\endgroup
%
\condition{P2}~
If there exists an execution of $\prog$ starting from the initial configuration $(\globalStore, \ThreadPool)$ that reaches the final configuration $(\globalStore', \emptyset)$,
then there also exists an execution of $\prog'$ from the same initial configuration, that either reaches the same final configuration or results in a failure:
\vspace{-2mm}
\begingroup
\small
\setlength{\jot}{2pt}
\begin{align*}
    (\globalStore, \ThreadPool) \stepp[\prog]\transitive (\globalStore', \emptyset) \implies (\globalStore, \ThreadPool) \stepp[\prog']\transitive  (\globalStore', \emptyset) \lor (\globalStore, \ThreadPool) \stepp[\prog']\transitive  \fail \\[-5mm]
\end{align*}
\endgroup
The refines relation is transitive, i.e., if $\prog_1 \refines \prog_2$ and $\prog_2 
\refines \prog_3$, then $\prog_1 \refines \prog_3$.

The notation $\stmt[\stmt_2 / \stmt_1]$ denotes the result of replacing all occurrences of statement $\stmt_1$ with statement $\stmt_2$ inside the statement $\stmt$.
Similarly, the notation $\prog[\stmt_2 / \stmt_1]$ denotes a new program in which, for every procedure $Q$,
all occurrences of $\stmt_1$ in the body of $Q$, as defined in $\prog$, are replaced with $\stmt_2$:
\begingroup
\small
\setlength{\jot}{2pt}
\begin{align*}
    \prog[\stmt_2 / \stmt_1] = \{Q \mapsto \prog(Q)[\stmt_2 / \stmt_1] \mid Q \in \ProcName\}
\end{align*}
\endgroup

A statement $\stmt$ is {\em terminating} in program $\prog$ if
$\prog$ does not have any infinite executions
from any configuration $(\globalStore, \Leaf{(\localStore, \stmt)})$ such that $\welltyped(\stmt, \LocalVar)$.
A well-typed program $\prog$ is {\em terminating} if for all $Q \in \dom(\prog)$ such that $\movertype(Q) \sqsubseteq \leftmover$,
we have $\prog(Q)$ is terminating in $\prog$.

\begin{theorem}
    \label{thm:maintheorem}
Let $\sourceprog$ be an atomic-free and well-typed program.
Let
\begingroup
\small
\setlength{\jot}{2pt}
\begin{align*}
    \interprog &= \sourceprog[\seqstmt{\stmt_1}{\stmt_2} \;/\; \parreduce{\stmt_1}{\stmt_2}]\\
    \reducedprog &= \interprog[\atomicstmt{\stmt} \;/\; \seqreduce{\stmt}]
\end{align*}
\endgroup
Then, the programs $\interprog$ and $\reducedprog$ are well-typed.
Furthermore, if $\interprog$ is terminating, then: 
(1) $\sourceprog \refines \interprog$, and 
(2) $\interprog \refines \reducedprog$. Therefore, $\sourceprog \refines \reducedprog$. 
\end{theorem} 

\noindent\textbf{[Proof Sketch]}
First, we establish that $\sourceprog \refines \interprog$, thereby showing that it is sound to sequentialize the concurrent behavior inside $\stmtfont{par\mhyphen reduce}$.
We prove refinement properties~\condition{P1} and~\condition{P2} separately. The top-level strategy is to rewrite an execution of $\sourceprog$ into an execution of $\interprog$ such that, for each application of the $\stmtfont{par\mhyphen reduce}$ rule, the statement $\stmt_1$ executes before $\stmt_2$, using an induction on the number of unreduced $\parreduceempty$ applications. Note that $\interprog$ eliminates all occurrences of $\parreduceempty$, including those nested inside $\seqreduceempty$.
As a consequence, there is no parallelism within any $\seqreduceempty$ application.

Second, we establish that $\interprog \refines \reducedprog$, thereby showing that it is sound to define atomic sections for all code blocks inside $\stmtfont{seq\mhyphen reduce}$. This step also proceeds by induction on the number of unreduced $\seqreduceempty$ applications. We show this by rewriting an execution of $\interprog$ into an execution of $\reducedprog$ in which each code block inside $\stmtfont{seq\mhyphen reduce}$ is of the form $\rightmover^* \nonmover? \leftmover^*$.

See Appendix~\ref{appendix:main} for more details.

\section{Implementation}
\vspace{-3mm}
\label{sec:implementation}
We have implemented our proof rule in \civl~\cite{DBLP:conf/fmcad/KraglQ21} verifier
for layered concurrent programs~\cite{DBLP:conf/cav/KraglQ18}.
Our implementation covers every aspect of our formalization except for the side conditions
on termination of left-mover procedures and absence of failures in right-mover statements.
For the examples reported in \autoref{sec:evaluation}, the verification of these
side conditions was done manually.

\civl is an extension of the Boogie verifier~\cite{DBLP:conf/fmco/BarnettCDJL05} for sequential programs.
Similar to Boogie, 
the implementation of \civl is broadly split into a type checker and a verification-condition generator.
The type checker handles basic type analysis and checks in addition that layer annotations on
variables, yield invariants, actions, and procedures are consistent with each other.
It also checks that the mover type of each procedure is consistent with the mover type inferred from
the body of the procedure.
The verification-condition generator in \civl eliminates all concurrency features from the input program
and produces a collection of sequential procedures annotated with specfications.
These sequential procedures encode checks related to mover types of atomic actions, refinement checks
for each procedure, and noninterference checks related to yield invariants.
The sequential procedures are processed by the standard Boogie flow that converts each procedure to a
logical constraint and checks it using an SMT solver.

Our implementation modifies and extends \civl as follows.
First, we modified the verification conditions generated for checking mover types of actions
according to the rules laid out for well-formed action map in \autoref{sec:movers}.
The revised rules are more general and therefore applicable in more scenarios.
The revised failure preservation check for left movers
provide an easier mental model while debugging unsuccessful proofs.

Second, we added mover types to procedures and implemented the checking of these types
against procedure bodies, as described in \autoref{sec:types}.
Our type checker also accounts for mover types of procedures in determining the degree of
program transformation allowed between successive program layers.
The ability to reduce programs in a fully nested manner, as described in \autoref{sec:refinement},
is important to enable a single layer to perform a large chunk of the proof, thus reducing the proof
overhead of layers.
We allow procedures annotated with mover types to be summarized using preconditions and postconditions.
Atomic code fragments with calls to such procedures can now be analyzed without inlining these procedures.
Our implementation also handles loops directly in the same manner as recursive~procedures.

Finally, we allow parallel calls to be reduced using the parallel reduction technique introduced
in this paper.
We provide parallel execution of procedure calls rather than parallel execution of statements.
This design choice simplified parameter passing between the caller and the callees.
The modified local variable analysis described in this paper is unnecessary and replaced by a simple
check that the output variables used across all arms of a parallel call are all distinct from each other.
Our implementation includes, in the same framework, the proof rule for synchronizing asynchronous calls
reported earlier~\cite{DBLP:conf/concur/KraglQH18}.

Our formalization of \lang in \autoref{sec:language} makes explicit every application of 
$\parreduceempty$ and $\seqreduceempty$ in the source program.
These annotations are not explicitly declared in a \civl program;
instead, our type checker automatically infers information equivalent to them.

\section{Evaluation}
\label{sec:evaluation}
\vspace{-2mm}
We evaluate the implementation described above on a diverse set of challenging case studies: a parallelized snapshot object (\autoref{sec:overview}), the classic message-passing simulation of shared memory by Attiya, Bar-Noy, and Dolev~\cite{DBLP:journals/jacm/AttiyaBD95} (ABD), an implementation of the FLASH cache coherence protocol~\cite{DBLP:conf/isca/KuskinOHHSGCNBHGRH94}, and a version of the Two-Phase Commit protocol. These implementations naturally decompose into procedures and make significant use of dynamic thread creation. Message passing is modeled in the style of RPC: broadcasting and waiting for responses is expressed as a parallel composition of procedure calls, each modifying the receiver's state and returning an acknowledgment.

\begin{wrapfigure}{l}{0.37\textwidth}
\captionsetup{type=table} 
\centering
{\small  

\vspace{-7mm}
\newcommand{\lbla}[1]{$\bf #1$}
\newcommand{\lblb}[2]{\hspace{-0.8mm}$\def\arraystretch{0.7}\begin{array}{c}\text{\scriptsize\bf #1}\\\text{\tiny #2}\end{array}$\hspace{-0.8mm}}

\begin{tabular}{lccc}
\toprule
\lbla{Example} &
\lblb{\#LOC}{Total} &
\lblb{\#LOC}{Impl \& Spec} &
\lblb{Time}{sec} 
\\ \midrule
Snapshot                    & 119  & 82 &   0.4\\
ABD                         & 389 &  206 &  1.4 \\
Coherence             &  608 & 401 &  5 \\
2PC                         & 146  &  111 &  1.8 \\
\bottomrule
\end{tabular}
\caption{Evaluation metrics.}
\label{tab:examples}
\vspace{-8mm}
}
\end{wrapfigure}

The evaluation shows that each case study involves substantial nesting of parallel and sequential reductions, leveraging both left and right mover types. This approach helps avoid the need for complex invariants that would otherwise arise from fine-grained interleavings. As common in Civl, the proofs are decomposed into a sequence of refinement steps, some of these steps being ``abstraction'' steps that are not related to reduction (they introduce non-deterministic abstractions or ghost variables). Also, these proofs rely on using the other features of Civl, e.g., inductive invariants and permission-based reasoning. The latter is particularly useful to enable commutativity reasoning. Many times, commutativity between actions is implied by the distinctness of (some of) their inputs and this is encoded using permissions.
 
~\autoref{tab:examples} presents quantitative metrics from our evaluation. We report the total lines of code for each proof, along with a separate count for the implementation and specification components. We also report that the wall-clock time required for executing each proof. The ratio of proof annotation lines to the combined lines of implementation and specification ranges from 0.31 (for Two-Phase Commit) to 0.88 (for ABD). Also, all these proofs can be completed in few seconds.
\edit{In addition, the Civl repository contains approximately 50 further examples, including larger benchmarks (on the order of hundreds to thousands of lines) where the seq-reduce rule applies. We focused on the four case studies in the paper because they incorporate both seq-reduce and par-reduce, and therefore most clearly illustrate the interaction and use of the two rules.}


In the following, we give more details about each case study, except for the snapshot object that we already described in \autoref{sec:overview}. We present the implementation and the specification we prove, and the use of reduction. The proof files are available in the supplementary material.

\subsection{The ABD register}

\begin{figure}
  \vspace{-8mm}
\centering
\begin{minipage}{.60\textwidth}
\begin{lstlisting}[numbers=none]
type TimeStamp // a set with a total order and a lower bound
TS: TimeStamp // global timestamp used to order operations
value_store: Map TimeStamp Value 

procedure ReadClient(pid) returns (val) {
  old_ts := Begin(pid)
  ts, val := Read(pid, old_ts)
  End(pid, ts);
}

procedure WriteClient(pid, val) {
  old_ts := Begin(pid)
  ts := Write(pid, old_ts)
  End(pid, ts);
}
\end{lstlisting}
\end{minipage}\hfill
\begin{minipage}{.41\textwidth}
\begin{lstlisting}[numbers=none]
action Begin(pid) returns (ts) {
  ts = TS;
}
action Read(pid, old_ts) returns (ts, val) {
  assume old_ts <= ts
  assume ts in value_store
  val := value_store[ts]
}
action Write(pid, val) returns (ts) {
  assume old_ts < ts
  assume ts not in value_store
  value_store[ts] := val
}
action End(pid, ts) {
  TS := max(TS, ts)
}
\end{lstlisting}
\end{minipage}
\vspace{-4mm}
\caption{A linearizable specification for ABD. The \code{TS} global variable models a clock which is read when operations start and advances when they end.}
\label{fig:ABDabstract}
\vspace{-7mm}
\end{figure}

The ABD algorithm implements a read-write register (shared memory) on top of message passing. It provides two operations \code{Read} and \code{Write} on a register that is replicated across $n$ replicas for fault tolerance. Less than half of the replicas can crash. The operations are invoked in parallel by a collection of clients.

Each replica stores a timestamped value (timestamp, value) where timestamp comes from a totally ordered set with a lower bound (e.g., natural numbers). Both \code{Read} and \code{Write} operations have two phases: the \code{QueryPhase} and the \code{UpdatePhase}. In the \code{QueryPhase}, they send \code{Query} messages to all replicas in parallel, wait for a quorum (at least half) of replies, and retrieve the reply \code{(t, v)} with the maximum timestamp \code{t} among the responses. Then they enter the \code{UpdatePhase}, where they send \code{Update} messages to all replicas. The \code{Read} operation sends an \code{Update(t, v)} message, while the \code{Write} operation sends an \code{Update(t+1, v')} where \code{v'} is the new value to be written. They both wait for a quorum of acknowledgements before returning. When receiving an \code{Update(t, v)} message, a replica updates its store to \code{(t, v)} if $t$ is greater than its current timestamp, and replies with an acknowledgment regardless of whether it updates. Upon receiving a \code{Query} message, it replies with its current copy. 

\noindent
\textbf{Specification: Linearizability.} Our goal is to prove that this register implementation is linearizable, which we reduce to a refinement check. To capture the condition that each operation appears to take effect atomically between its call and return, we instrument operations to read a global clock at the start, which advances when they end. This clock is aligned with ABD timestamps: each \code{Read} returns a value with a timestamp greater than or equal to the clock at invocation, and each \code{Write} writes a value with a timestamp strictly greater than the clock at invocation. At the end of an operation, the clock advances as the timestamp read/written by that operation. The resulting specification is shown in \autoref{fig:ABDabstract}, where \code{ReadClient} and \code{WriteClient} wrap the respective operations with \code{Begin} and \code{End} actions for clock access. Each method contains a single ``internal'' step (a call to the action \code{Read} or \code{Write}), reflecting the requirement that they take effect instantaneously. The map \code{value\_store}, updated by \code{Write} and accessed by \code{Read}, ensures that reads return previously written values.

We prove that the ``concrete'' versions of \code{ReadClient} and \code{WriteClient} where the calls to the actions \code{Read} or \code{Write} are replaced by calls to the homonymous ABD procedures are a refinement of the abstract specification in \autoref{fig:ABDabstract}.

\vspace{.5mm}
\noindent
\textbf{Applying reduction.} The goal is to apply reduction to show that the ABD procedures \code{Read} and \code{Write} can be rewritten to execute within a \emph{single} atomic section, which is then shown to refine the  abstract \code{Read} and \code{Write} actions in \autoref{fig:ABDabstract}. To achieve this, we introduce an abstraction of the \code{Query} handler which is a right mover, which in turn ensures that the \code{Update} handler becomes a left mover. This abstraction enables parallel reductions in both the \code{QueryPhase} and \code{UpdatePhase}; for example, in the \code{QueryPhase} shown below, it allows the recursive procedure to be reduced by sequentializing all \code{Query} handlers (which initially happened in parallel). This sequentialization, in turn, enables a sequential reduction within the \code{Read} and \code{Write} procedures—illustrated below for \code{Read}—since the \code{Query} handlers (right movers) are followed by \code{Update} handlers (left movers).

\begin{figure}[t]
\centering

\begin{minipage}[t]{0.54\textwidth}
\begin{lstlisting}
right procedure {:layer 3} QueryPhase(i: int, old_ts: TimeStamp)
  returns (max_ts: TimeStamp, max_value: Value) {
  ...
  par-reduce {
    call max_ts, max_value :=
      QueryPhase(i + 1, old_ts)
    par call ts, value :=
      Query(i, old_ts)  // right
  }
  if (less_than(max_ts, ts)) {
    max_ts := ts; max_value := value;
  }
}
\end{lstlisting}
\end{minipage}\hfill
\begin{minipage}[t]{0.44\textwidth}
\begin{lstlisting}
procedure Read(pid: ProcessId, old_ts: TimeStamp) returns (ts: TimeStamp, value: Value) {
  ...
  seq-reduce {
    call ts, value :=
      QueryPhase(0, old_ts);   // right
    call UpdatePhase(0, ts, value); // left
  }
}
\end{lstlisting}
\end{minipage}
\vspace{-7mm}
\caption{$\texttt{QueryPhase}$ and $\texttt{Read}$ procedures.}
\vspace{-9mm}
\end{figure}

The abstraction of \code{Query} allows it to return a timestamp lower than the current replica timestamp. However, it cannot return any arbitrarily low timestamp---it should be greater than or equal to the clock timestamp \code{old\_ts} obtained in \code{Begin}. This reduction eliminates the need for an inductive invariant that tracks relationships between read and write operations across different stages of their query or update phases.

\subsection{Cache coherence}

We implement the FLASH cache coherence protocol ~\cite{DBLP:conf/isca/KuskinOHHSGCNBHGRH94} which is a directory-based MESI cache coherence protocol. The protocol manages consistency across multiple caches in a shared memory multiprocessor system using a centralized directory. 

\autoref{fig:coherenceVariables} shows our model for the memory (\code{mem}), directory (\code{dir}) and a set of caches (\code{cache}). 
The memory is indexed by memory addresses (\code{MemAddr}), while each cache uses local cache addresses (\code{CacheAddr}) for indexing. Because the memory address space is larger than the cache address space, a hash function maps each memory address to a cache address, allowing multiple memory addresses to correspond to the same cache address. Each cache stores data in cache lines, which contain a memory address, a value, and a state (\code{Modified, Exclusive, Shared, Invalid}). The directory tracks the status of each memory address across all caches. If a memory address is held in the Modified or Exclusive state by a cache, the directory records this as \code{Owner(i)}, where \code{i} is the ID of the owning cache. Otherwise, the directory records it as \code{Sharers(iset)}, where \code{iset} is the set of cache IDs having the memory address in Shared state.

\begin{figure}[t]
\centering

\begin{minipage}[t]{0.54\textwidth}
\begin{lstlisting}
type MemAddr; type CacheAddr;
function Hash(MemAddr): CacheAddr;
datatype State
  {Modified(), Exclusive(), Shared(), Invalid()}
datatype CacheLine
  {CacheLine(ma: MemAddr, value: Value, state: State)}
datatype DirState
  {Owner(i: CacheId), Sharers(iset: Set CacheId)}
\end{lstlisting}
\end{minipage}\hfill
\begin{minipage}[t]{0.44\textwidth}
\begin{lstlisting}
// Implementation state
var {:layer 0,2} mem: [MemAddr]Value;
var {:layer 0,2} dir: [MemAddr]DirState;
var {:layer 0,2} cache: [CacheId][CacheAddr]CacheLine;
// Specification state
var {:layer 1,3} absMem: [MemAddr]Value;
\end{lstlisting}
\end{minipage}

\vspace{-2mm}
\caption{State representation for cache coherence protocol}
\label{fig:coherenceVariables}
\vspace{-4mm}
\end{figure}

We implement 5 top-level operations on the cache: 
\begin{itemize}
    \item \code{cache\_read} and \code{cache\_write} which read and write a cache entry, respectively. 
    \item \code{cache\_evict\_req} initiates eviction of a cache line.
    \item \code{cache\_read\_shd\_req} and \code{cache\_read\_exc\_req} initiate bringing a memory address into the cache in Shared and Exclusive mode, respectively.
\end{itemize}

\vspace{-2mm}
We now detail the operation of \code{cache\_read\_exc\_req} and the interactions between cache and directory. The cache initiates an exclusive request to the directory via \code{dir\_read\_exc\_req} shown in \autoref{fig:dir_read_exc_req}. If the directory state for the requested memory address is \code{Owner}, it sends an invalidate request to the owner by calling \code{cache\_invalidate\_exc}. The owner is then expected to change its state to \code{Invalid} and send the data back to the directory, which writes it to memory by calling \code{write\_mem}. If the directory state is \code{Sharers}, it sends invalidate requests in parallel to all caches in the sharers list by invoking \code{cache\_invalidate\_shd}. The directory then reads the memory by calling \code{read\_mem} and sends the response back to the orginal cache via \code{cache\_read\_resp}, and ending the request by calling \code{dir\_req\_end}.

\begin{figure}[t]
\begin{minipage}{.58\textwidth}
\begin{lstlisting}
procedure dir_read_exc_req(i: CacheId, ma: MemAddr)
{
 ... // variable initialization
 seq-reduce {
   call dirState := dir_req_begin(ma); // right
   if (dirState is Owner) {
     call value := cache_invalidate_exc
                      (dirState->i, ma, Invalid()) // non-mover
     call write_mem(ma, value); // both mover
   } 
   else {
     call dp := invalidate_sharers(ma, dirState->iset); // left
     call value := read_mem(ma); // both mover
   }
   call cache_read_resp(i, ma, value, Exclusive()); // left
   call dir_req_end(ma, Owner(i)); // left
 }
}
\end{lstlisting}
\end{minipage}
\hspace{-4mm}
\begin{minipage} {.4\textwidth}
\begin{lstlisting}
left procedure invalidate_sharers
(ma: MemAddr, victims: Set CacheId) {
  ...
  if (victims == Set_Empty()) {
    return;
  }
  victim := Choice(victims->val);
  victims' := Set_Remove(victims, victim);
  par-reduce { // left
   call cache_invalidate_shd(victim, ma, Invalid()) 
   par call invalidate_sharers(ma, victims') 
  }
}
\end{lstlisting}
\end{minipage}
\vspace{-4mm}
\caption{Reduction applied at directory for exclusive state request}
\label{fig:dir_read_exc_req}
\vspace{-7mm}
\end{figure}

\smallskip
\noindent
\textbf{Specification.} We introduce an abstract memory, \code{absMem}. The goal is to show that the \code{cache\_write} and \code{cache\_read} operations are refinement of atomic actions that directly read and write from \code{absMem}. We hide directory and all cache operations that interact with it. This specification naturally captures the cache coherence property.

\smallskip
\noindent
\textbf{Applying reduction.} After introducing the ghost variable \code{absMem} used in the specification, we define a number of abstractions of actions used in the implementation that become movers. The memory operations \code{read\_mem} and \code{write\_mem} are made both movers, the shared invalidate request \code{cache\_invalidate\_shd} is made a left mover, the response to a read request at a cache \code{cache\_read\_resp} is made a left mover, the actions \code{dir\_req\_begin} and \code{dir\_req\_end} for reading and updating the directory state are made right and left movers, respectively. 
These movers enable reduction at many sites in the implementation. In particular, they enable parallel reduction in the invalidate loop (\code{invalidate\_sharers}) to sequentialize them, and subsequently, a sequential reduction on the entire body of the procedure \code{dir\_read\_exc\_req} for bringing a memory address into the cache in Exclusive mode. This is made precise in \autoref{fig:dir_read_exc_req}.

The reduction helps a refinement proof to hide the directory and all the caches so that the read and write operations at cache are refinements of the atomic operations over \code{absMem}. 

\subsection{Two-phase Commit (2PC)}
\vspace{-1mm}
Two-phase Commit is a classic distributed protocol used to implement concurrent transactions. A number of \emph{coordinator} processes make a number of \emph{replicas} agree on an order between concurrent transactions. Each transaction is associated with a start time and an end time, and it is submitted to a single coordinator. Two transactions conflict if their time intervals overlap. The goal is to ensure that all committed transactions are not conflicting pairwise. 

A coordinator runs in two phases. In the vote phase, it sends vote requests to all replicas which reply with \code{YES} or \code{NO} (accept or not a transaction). A replica stores the set of pending transactions (which are not yet committed or aborted) for which it already voted \code{YES} in a so-called \emph{locked set}, and it answers \code{YES} iff the incoming vote request concerns a transaction that does not conflict with some transaction in the locked set. In the finalize  phase, if all replies are \code{YES}, the coordinator sends a commit request and otherwise, an abort request. If a replica receives an abort request, it removes the transaction from the locked set.

\noindent
\textbf{Specification.} We add a ghost variable, \code{committed\_transactions}, which keeps track of all transactions that have been committed. Before adding a transaction to this set (in the coordinator's code), we assert that it does not conflict with any previously committed transaction. 

\noindent
\textbf{Applying reduction.} To enable reduction, we abstract the vote request handler to allow it to non-deterministically respond with \code{NO} without modifying the state. This abstraction makes the vote request handler a right mover, while the abort request handler becomes a left mover. The commit request handler is a both mover, since it does not change the state.
To illustrate, consider two successive vote requests handled by the same replica (requests at different replicas commute, as they access disjoint state). If the transactions conflict, one handler might return \code{YES} and the other \code{NO}. Without the abstraction, reordering these handlers isn't sound: if the second executes first, it might respond \code{YES}, which breaks commutativity. However, with the abstraction, the second handler can non-deterministically return \code{NO}, allowing the reordered execution where the first still responds \code{YES}. Similar reasoning applies to other combinations of handlers.

As in previous cases, the abstraction enables a combination of parallel reduction and sequential reduction. The parallel reduction is used to sequentialize the two phases, as exemplified below for the vote requests, and then, the sequentialization enables showing that the whole computation for a transaction can be executed within an atomic section.

\begin{figure}[h]
  \vspace{-5mm}
\begin{minipage}{.55\textwidth}
\begin{lstlisting}
right procedure vote_all(xid: TransactionId, i: ReplicaId) 
  returns (votes: [ReplicaId]Vote) {
  ...
  if (1 <= i) {
    vr = VoteRequest(xid, i);
    par-reduce { 
     (call votes := vote_all(xid, i-1)) 
     par (call out := vote(vr)); } // right 
    }
    votes[i] := out;
}}
\end{lstlisting}
\end{minipage}
\begin{minipage}{.4\textwidth}
\begin{lstlisting}
procedure TPC(xid: TransactionId) {
  ...
  seq-reduce {
    call votes := vote_all(xid, n); // right 
    // locally calculate decision based on votes
    call finalize_all(decision, xid); // left 
  }
  ...
}
\end{lstlisting}
\end{minipage}
\vspace{-5mm}
\caption{$\code{vote\_all}$ and $\code{TPC}$ procedures.}
\vspace{-11mm}
\end{figure}

\section{Related Work}\label{sec:related}
\vspace{-1mm}
We review works concerning the use of commutativity reasoning in proving correctness of concurrent or distributed systems. 

\vspace{.5mm}
\noindent
\textbf{Commutativity reasoning in deductive verification.}
Lipton's reduction theory~\cite{DBLP:journals/cacm/Lipton75} introduced the concept of \emph{movers}
to define a program transformation that creates bounded-size atomic blocks. This work assumes a simple programming language without procedure calls and a fixed number of threads.
QED~\cite{DBLP:conf/popl/ElmasQT09} expanded the scope of Lipton's theory by introducing iterated
application of reduction and abstraction over atomic actions. Also, atomic sections are allowed to contain loops but no procedure calls or dynamic thread creation. 
Civl~\cite{DBLP:conf/cav/HawblitzelPQT15} builds upon the foundation of QED,
adding invariants~\cite{DBLP:journals/cacm/OwickiG76,DBLP:conf/ifip/Jones83},
refinement layers and permission-based reasoning via a linear type system~\cite{DBLP:conf/cav/KraglQ18}, and pending asyncs~\cite{DBLP:conf/concur/KraglQH18,DBLP:conf/pldi/KraglEHMQ20}. Pending asyncs can be viewed as threads restricted to executing a single atomic step and which cannot be joined. They are used to summarize asynchronous procedure calls and define a reduction scheme where asynchronous procedure calls are transformed to synchronous ones~\cite{DBLP:conf/concur/KraglQH18}. This reduction scheme is based on proving that the asynchronously called procedure can be summarized to a left-mover pending async. This idea has been extended to sequentializing an asynchronous program that creates an unbounded number of pending asyncs via an induction principle~\cite{DBLP:conf/pldi/KraglEHMQ20}. In this work, we introduce more flexible reduction schemes that improve scalability. These schemes support greater compositionality by allowing atomic sections to include both sequential and parallel procedure calls. Additionally, they expand the capabilities of reduction by enabling both left- and right-mover-based commutative reorderings.

Anchor~\cite{DBLP:journals/pacmpl/FlanaganF20} applies reduction to a low-level object-oriented language, where mover annotations are assigned to read and write accesses to object fields. It introduces a type system that enables proving the atomicity of entire procedures, which builds on Lipton's reduction. In contrast, our work is set in a more abstract language, supports compositional reduction reasoning about procedures, and accounts for parallel composition. \cite{DBLP:conf/ecoop/FlanaganF24} investigates the integration of reduction with rely-guarantee reasoning, which falls outside the scope of this work.

CSPEC~\cite{DBLP:conf/osdi/ChajedKLZ18} takes an approach similar to Civl but mechanizes all metatheory
within the Rocq theorem prover~\cite{Coq} for flexibility and sound extensibility.
Armada~\cite{DBLP:conf/pldi/LorchCKPQSWZ20} also has flexible and mechanized metatheory whose usefulness is demonstrated by
implementing a variety of program transformations, including those catering to fine-grained concurrency
and weak memory models. IronFleet~\cite{DBLP:conf/sosp/HawblitzelHKLPR15} embeds TLA-style state-machine modeling~\cite{DBLP:books/aw/Lamport2002} into the Dafny verifier~\cite{DBLP:conf/lpar/Leino10} to refine high-level distributed systems specifications into low-level executable implementations. Their proofs embed reduction reasoning into Dafny in a rather ad-hoc manner.

Movers have also been used to define an equivalence-preserving transformation that eliminates buffers in message-passing programs~\cite{DBLP:journals/pacmpl/BakstGKJ17,DBLP:journals/pacmpl/GleissenthallKB19}.
These works define a restricted class of programs and prove that reasoning about the set of \emph{rendezvous} executions of these programs, where messages are delivered instantaneously, is complete, i.e., any other execution is equivalent to a rendezvous execution, up to reordering of mover actions. For instance, \cite{DBLP:journals/pacmpl/GleissenthallKB19} introduces some number of heuristics which are based on syntax in order to reduce a given program. Those heuristics do not apply to our case studies, and it is hard to imagine an extension where they would become applicable. For instance, reduction is sometimes enabled by abstracting actions (message handlers) and this cannot be handled via syntactical arguments.
\edit{Two-phase commit (2PC) is a canonical benchmark in this line of work and has been verified many times in a variety of systems, including automated ones\cite{10.1145/2914770.2837650,DBLP:journals/pacmpl/GleissenthallKB19}.  
In our 2PC, replicas use nontrivial logic to determine their votes, which is not the case for the versions used in these systems. In those previous works, replicas vote \emph{Yes} or \emph{No} nondeterministically, which significantly simplifies the correctness argument: all message handlers are left movers without requiring any abstraction~\cite{DBLP:conf/concur/KraglQH18}.  
In contrast, in our version of 2PC, some message handlers are right movers and some are left movers, after devising appropriate abstractions.}

\vspace{.5mm}
\noindent
\textbf{Commutativity reasoning in algorithmic verification.}
In the context of algorithmic verification, commutativity reasoning manifests in the so-called partial-order reduction techniques~\cite{DBLP:books/sp/Godefroid96,DBLP:conf/popl/FlanaganG05,DBLP:conf/popl/AbdullaAJS14,DBLP:journals/pacmpl/Kokologiannakis22} which mostly concern finite-state systems or executions of bounded length.

In the context of automated proof synthesis for infinite-state programs, most existing work focuses on programs with a bounded number of threads~\cite{DBLP:conf/hvc/ChuJ14,DBLP:conf/pldi/FarzanKP22,DBLP:journals/pacmpl/FarzanKP23}. The work in~\cite{DBLP:journals/pacmpl/FarzanKP24} proposes an instrumentation scheme for parameterized programs, where an unbounded number of threads execute the same code. This scheme enables the representation of sound reductions in such settings. Additionally, they formalize a notion of reduction usefulness, suggesting that a suitable reduction can lead to proofs requiring fewer or simpler ghost variables.

\begin{credits}
\subsubsection{\ackname} This work is partially partially supported by the French National Research Agency (project SCEPROOF).
\end{credits}

\newpage
\bibliography{dblp}

\newpage
\appendix
\section{Proofs for Section 6 (Reduction for RedPL Programs)}\label{appendix:main}


\subsection{Parallel reduction}
\label{sec:proofsection}
In this section we provide the first part of the proof of \autoref{thm:maintheorem}.

We show that all occurrences of $\stmtfont{par \mhyphen reduce}$ can be sequentialized, i.e., $\sourceprog$ refines $\interprog$. We prove that the two properties~\condition{P1} and~\condition{P2} that define refinement are satisfied.


\paragraph*{\bf{Showing refinement property~\condition{P1}}}
Given a failing execution $\pi$ of $\sourceprog$ we rewrite it to $\pi'$, a failing execution in $\interprog$, such that, in the rewritten execution, every instance of $\parreduce{\stmt_1}{\stmt_2}$ satisfies the following two requirements:
\begin{enumerate}
    \item Ordering: all transitions of $\stmt_1$ must come before all transitions of $\stmt_2$.
    \item Completion: if there is a transition of $\stmt_2$, then $\stmt_1$ must have executed completely.
\end{enumerate}


The idea is to first satisfy the two requirements for instances which don't have any further par-reduce statements nested inside them.  
After that, we proceed to satisfy instances where the nested instances have already been satisfied, and so on until all instances satisfy both requirements.

Given an execution, we can identify an instance of a par-reduce applicaton which does not have any further nested applications as follows: in the resulting configuration of a par-reduce enter transition, place a pointer on the internal node $(\EmptyNode^*)$ that was introduced. 
If, in the remainder of the execution, there is no descendant of this node $(\EmptyNode^*)$ where a par-reduce enter transition is fired, then this instance contains no further nested par-reduce applications.

Let $\parreduce{\stmt_1}{\stmt_2}$ be such an instance. For this to be a well typed statement, the following condition related to mover type: $\movertype(\stmt_1) \sqsubseteq \leftmover \vee (\movertype(\stmt_2) \sqsubseteq ~\rightmover \land \neg \mayfail(\stmt_2))$ needs to hold.
First, let's consider the case when $\movertype(\stmt_1) \sqsubseteq \leftmover$, we call this case left-reduce.

\vspace{-2mm}
\paragraph*{\bf{Satisfying ordering requirement (left-reduce)}}
In addition to the ordering requirement, in the case of left-reduce, we also need transitions of $\stmt_1$ to be interference-free (occur consecutively) to ensure that our completion process (explained later) terminates. Therefore, in this case while ordering we also remove interference within transitions of $\stmt_1$.

In order to satisfy the ordering requirement, we start from the leftmost transition of $\stmt_1$ that is not in order and attempt to move it stepwise left (possibly replacing it with a different transition in the process) until it is in the correct position. 
We continue until all transitions of $\stmt_1$ are in order, each time picking the next leftmost transition of $\stmt_1$ that is not yet in order and positioning it immediately after the previously ordered transition.

In the process of moving stepwise to the left, the only interesting case arises when two action transitions are adjacent, otherwise, swapping is trivial because the transitions being swapped do not affect each other. 
Therefore, we only look at this scenario.
Let $L = ( \blank , \blank, \gate_L, \blank, \MoverVar_L, \blank)$ be the action whose transition ($t_L$) we are trying to move to the left, and let $X$ be the action whose transition ($t_X$) immediately precedes $t_L$. 
Let $\ell_L= \inputbinding(t_L), \ell_X = \inputbinding(t_X)$ be their input bindings. Let $\gate_L$ and $\gate_X$ be their respective gates. 

Note that since $\movertype(\stmt_1) \sqsubseteq \leftmover$, action $L$ must have left-mover properties i.e. $\MoverVar_L \sqsubseteq \leftmover$.

Our goal is to reorder the execution while preserving the failure outcome from the initial state; preservation of intermediate states is not required, and the reordered execution may be shorter.
~\\
\textbf{case 1}: $t_L$ is a failure transition: 
\[
(\globalStore_1, \ThreadPool_1) \xrightarrow{t_X} (\globalStore_2, \ThreadPool_2) \xrightarrow{t_L} \fail
\]
We know from $\preservesSuccess(X, L)$ that whenever gate of $X$ and gate of $L$ hold in a state, then any transition of $X$ leads to a state where gate of $L$ holds. Using the contrapositive of this condition, we get: if from $(\globalStore_1, \ThreadPool_1)$ we have a transition of $X$ to a state $(\globalStore_2, \ThreadPool_2)$ where the gate of $L$ does not hold, then either gate of $X$  or gate of $L$ does not hold in $(\globalStore_1, \ThreadPool_1)$. Since, in our case $\globalStore_1\cc\ell_X \in \gate_X$, then it must be that  $\globalStore_1\cc\ell_L \notin \gate_L$. Therefore, we can get a failure of $L$ from $(\globalStore_1, \ThreadPool_1)$ and $t_X$ can be eliminated from the execution. 
~\\
\textbf{case 2}: $t_L$ is a non-failure transition: 
\[
(\globalStore_1, \ThreadPool_1) \xrightarrow{t_X} (\globalStore_2, \ThreadPool_2) \xrightarrow{t_L} (\globalStore_3, \ThreadPool_3)
\]
\begin{itemize}
    \item[] \textbf{case 2.1}:  $(\globalStore_1, \ell_X, \ell_L) \in wlp(X, \gate_L)$ and $(\globalStore_1, \ell_L, \ell_X) \in wlp(L, \gate_X)$.
    \\There exists a way to execute $L$ followed by $X$ and reach the state $(\globalStore_3, \ThreadPool_3)$ by $\commutes(X, L)$~\condition{L3} property of left-movers, thereby preserving the original failure in $\pi$.

    \item[] \textbf{case 2.2}: $(\globalStore_1, \ell_X, \ell_L) \notin wlp(X, \gate_L)$. 
    \\We already have $\globalStore_1\cc\ell_X \in \gate_X$, which means there is a way to run $X$ from $(\globalStore_1, \ThreadPool_1)$ and get a failing transition of $L$. Therefore, we have a transition of $X$ followed by a failure transition of $L$, which is the same as case 1.

    \item[] \textbf{case 2.3}: $(\globalStore_1, \ell_L, \ell_X) \notin wlp(L, \gate_X)$. 
    \\If $\globalStore_1\cc\ell_L \notin \gate_L$, then we execute L from $(\globalStore_1, \ThreadPool_1)$ and get a failure, eliminating $t_X$, otherwise if $\globalStore_1\cc\ell_L \in \gate_L$ then there is a way to run $L$ from $(\globalStore_1, \ThreadPool_1)$ and get a failing transition of $X$. 

\end{itemize}

Let $\pi_l$ be the resulting execution where all transitions of $\stmt_1$ are ordered and interference-free.

\vspace{-2mm}
\paragraph*{\bf{Satisfying completion requirement (left-reduce)}}
Suppose that in $\pi_l$, $\stmt_1$ did not execute completely, and there is at least one transition of $\stmt_2$. In this case, we need to complete $\stmt_1$. Let $\seq(s) = s[\stmt';\stmt'' \;/\; \parreduce{\stmt'}{\stmt''}]$. To complete $\stmt_1$ we choose a scheduling order that matches $seq(\stmt_1)$. The unexecuted part of $\stmt_1$ may contain par-reduce statements, but none have started (otherwise we would have nested par-reduce). Thus, the executed transitions of $\stmt_1$ in $\pi_l$ match a prefix of $seq(\stmt_1)$, and we continue by extending the execution of $\stmt_1$ such that its transitions always match a prefix of $\seq(\stmt_1)$.

Let the last state of the execution before the failure be $(\globalStore, \ThreadPool)$ and let $t_X$ be the failure transition of action $X$ from this state. From $(\globalStore, \ThreadPool)$,  we allow $\stmt_1$ to execute and take a transition that is consistent with the scheduling order of $\seq(\stmt_1)$. 

If this is a non-action transition, it can be trivially introduced while preserving the failure. Therefore, we assume it is an action transition $t_L$ corresponding to action $L$. When we introduce it, one of the two following cases may happen:

\textbf{case 1}: $t_L$ results in a failure. 

Then we have preserved the failure. 

\textbf{case 2}: $t_L$ does not fail. 

By the $\preservesFailure(L, X)$~\condition{L2} property of left-mover we have that if $X$ fails from a state that $L$ does not, then there is a way to execute $L$ and reach a failure of $X$. We use this property to execute $L$ from $(\globalStore, \ThreadPool)$ and reach a state from where $X$ fails. 

Having introduced $t_L$ while preserving the failure, we can now order it to be in its correct position in the same way as before. 
We continue until $\stmt_1$ has executed completely or we have eliminated all transitions of $\stmt_2$ in the ordering process.

\vspace{-2mm}
\paragraph*{\bf{Termination of completion}} Here is the reason why this process of completing $\stmt_1$ terminates. 
The only possible source of non-termination inside $\stmt_1$ comes from procedure calls. 
Let's say there is a call to some procedure $Q$ inside $\stmt_1$. Since $\stmt_1 \sqsubseteq \leftmover$, it follows that $\movertype(Q) \sqsubseteq \leftmover$.
From the assumption that $\interprog$ is terminating, we know that $\interprog(Q)$ terminates when run in isolation (i.e., without interference) starting from an arbitrary store with input bindings. Given that $\seq(\sourceprog(Q)) = \interprog(Q)$, the execution of the procedure call $Q$ while completing $\stmt_1$ also terminates. 

Thus, the resulting execution satisfies the ordering and completion requirements for this instance.

Now we consider the case when $\movertype(\stmt_2) \sqsubseteq \rightmover$. We call this case right-reduce. 

\vspace{-2mm}
\paragraph*{\bf{Satisfying ordering requirement (right-reduce)}}
In order to satisfy the ordering requirement, our strategy is to start with the rightmost transition of $\stmt_2$ that is out of order and move it stepwise to the right until it is in the correct position. 
We then continue by selecting the next rightmost out-of-order transition, repeating this process until all transitions are correctly ordered.

In the process of moving stepwise to the right, again the only interesting case arises when two action transitions are adjacent. 
Let $R = ( \blank , \blank, \gate_R, \blank, \MoverVar_R, \blank)$ be the action whose transition ($t_R$) we are trying to move to the right, and let $X$ be the action whose transition ($t_X$) immediately follows $t_R$.
Let $\ell_R= \inputbinding(t_R), \ell_X = \inputbinding(t_X)$ be their input bindings and $\gate_R$ and $\gate_X$ be their gates. 

Since $\movertype(\stmt_2) \sqsubseteq \rightmover$, action $R$ must have right-mover properties i.e., $\MoverVar_R \sqsubseteq \rightmover$.
~\\
\textbf{case 1:} $t_X$ is a failure transition 
\[
(\globalStore_1, \ThreadPool_1) \xrightarrow{t_R} (\globalStore_2, \ThreadPool_2) \xrightarrow{t_X} \fail
\]
In this case, using the $\preservesSuccess(R, X)$ $\condition{R1}$ property of right movers in the contrapositive form (similar to how we used it in the case 1 of left-reduce) we can cause a failure of X from $(\globalStore_1, \ThreadPool_1)$, thereby eliminating the transition of action R. 
~\\
\textbf{case 2:} $t_X$ is a non-failure transition
\[
(\globalStore_1, \ThreadPool_1) \xrightarrow{t_R} (\globalStore_2, \ThreadPool_2) \xrightarrow{t_X} (\globalStore_3, \ThreadPool_3)
\]
\begin{itemize}
    \item[] \textbf{case 2.1}: $(\globalStore_1, \ell_X, \ell_R) \in wlp(X, \gate_R)$ and $(\globalStore_1, \ell_R, \ell_X) \in wlp(R, \gate_X)$.
    \\Then, there is a way to execute $X$ followed by $R$ and reach the state $(\globalStore_3, \ThreadPool_3)$ using $\commutes(R,X)$ \condition{R2}, thereby preserving the original failure in $\pi$.

    \item[] \textbf{case 2.2}: $(\globalStore_1, \ell_X, \ell_R) \notin wlp(X, \gate_R)$. 
    \\From $\neg \mayfail(\stmt_2)$ we know that action $R$ cannot fail, so we have $wlp(X, \gate_R) = \gate_X$. 
    It follows that $\globalStore_1\cc\ell_X \notin \gate_X$, so executing X from $(\globalStore_1, \ThreadPool_1)$ results in a failure of $X$, thereby eliminating the transition of $R$.

    \item[] \textbf{case 2.3}: $(\globalStore_1, \ell_R, \ell_X) \notin wlp(R, \gate_X)$. 
   \\From $\neg \mayfail(\stmt_2)$ we have that every state satisfies $\gate_R$, so $(\globalStore_1 \cc \ell_R) \in \gate_R$. Therefore, there must be a way to execute $R$ from $(\globalStore_1, \ThreadPool_1)$ and generate a failure of $X$. 
   Now, we fall into case~1.
\end{itemize}

Let $\pi_r$ be the resulting execution where all transitions of $\stmt_1$ are ordered.

\vspace{-2mm}
\paragraph*{\bf{Satisfying completion requirement (right-reduce)}}
Suppose in $\pi_r$, $\stmt_1$ is incomplete and there is at least one transition of $\stmt_2$. 
Since $\stmt_1$ may be a non-mover, we cannot complete it by extending the execution in a way that preserves the failure. 
Therefore, our strategy is to eliminate all transitions of $\stmt_2$ from the execution to satisfy the completion requirement. 
Note that if the failure transition belonged to $\stmt_2$, then we cannot eliminate transitions of $\stmt_2$. 
This is the reason we have the requirement in type checking that $\neg \mayfail(\stmt_2)$.
Now we can assume that the failure transition does not belong to $\stmt_2$. 
Starting from the last transition of $\stmt_2$ we move it stepwise right until we have eliminated it. 
This is achieved by applying the ordering process, as before, with the correct position set immediately after the failure transition. 
We continue until all transitions of $\stmt_2$ have been eliminated.

Now we have an execution in which the ordering and completion requirement of a single instance of par-reduce application (without any further nested par-reduce applications) has been satisfied. 
We can now proceed to adjust other such instances similarly, including the ones where the nested instances have already been processed. 
Once an instance has been adjusted to satisfy the ordering and completion requirements, these properties continue to hold despite the adjustment of any other instance. 
Here is why that is the case:
\begin{enumerate}
    \item The ordering requirement cannot be violated because, while fixing one particular instance, we never swap two transitions that both belong to another instance. 
    Therefore, the relative order among transitions of other instances remains unchanged.
    \item Regarding completion, the only concern may be that, while fixing another instance, the execution might be shortened by introducing a failure. As a result, we might lose a suffix of some already fixed instance. 
    However, this does not violate the completion requirement,  because the completion requirement holds for all prefixes of instances that satisfy ordering and completion. 
\end{enumerate}
As a result, when all instances have been processed, all instances satisfy both requirements.

\vspace{-2mm}
\paragraph*{\bf{Showing refinement property~\condition{P2}}}
Given an execution $\pi$ of $\sourceprog$ from an initial configuration $(\globalStore, \ThreadPool)$ that ends in a final configuration $(\globalStore', \emptyset)$, we aim to show that there exists an execution $\pi'$ in $\interprog$ that starts from $(\globalStore, \ThreadPool)$ and ends in the same final configuration or fails.
\begin{align*}
    (\globalStore, \ThreadPool) \stepp[\prog]\transitive (\globalStore', \emptyset) \quad \implies \quad \big((\globalStore, \ThreadPool) \stepp[\prog']\transitive  (\globalStore', \emptyset) \quad \lor \quad (\globalStore, \ThreadPool) \stepp[\prog']\transitive  \fail \big)
\end{align*}

We can establish this by applying the same rewrite process to enforce the required ordering on $\pi$. Completion is needed only if the reordering introduces failures. If no failures are introduced---that is, we remain in case~2.1 of both left- and right-reduce---then the final state $(\globalStore', \emptyset)$ is preserved by the $\commutes$ property of left and right movers. Otherwise, the execution ends in a failure.


\subsection{Sequential reduction}
Now we show the second part of the proof of \autoref{thm:maintheorem}.

We show that all code blocks inside $\stmtfont{seq \mhyphen reduce}$ can be made atomic, i.e., $\interprog$ refines $\reducedprog$.

\paragraph*{\bf{Showing refinement property~\condition{P1}}}
Given a failing execution $\pi$ of $\interprog$, we rewrite it to a failing execution $\pi'$ of $\reducedprog$ such that, in the rewritten execution, for each sequential reduce application of the form $\seqreduce{\stmt}$, the following two requirements hold:
\begin{enumerate}
    \item Non-interference: There is no interleaving of other transitions with those of statement~$\stmt$.
    \item Completion: If there is at least one transition of statement $\stmt$ in the rewritten execution, then statement $\stmt$ must either execute completely or end in a failure.
\end{enumerate} 

In order to satisfy non-interference and completion requirement of each instance, our strategy is to satisfy these two requirements for every outermost instance (i.e., an instance that is not nested within any other sequential reduce application). 
Once the outermost instance satisfies both requirements, all nested instances do so as well, since their transitions are treated as part of the outermost instance.

We can identify an outermost instance of the sequential reduction application as follows: in the resulting configuration of a seq-reduce enter transition, we mark the leaf node ($\EmptyLeaf^*$) where $\inseqreduceempty$ was introduced. 
If there is no $\inseqreduceempty$ in this node and any node on the path from $\EmptyLeaf^*$ to the root, then we have identified an outermost application.

\vspace{-2mm}
\paragraph*{\bf{Satisfying non-interference requirement}} We devise a swapping strategy on transitions by partitioning them into blocks and performing swaps based on the block type, such that when no further swaps are possible, the resulting execution satisfies the non-interference requirement.

\vspace{-2mm}
\paragraph*{\bf{Blocks}} We consider all transitions belonging to a statement $\stmt$ of a particular outermost instance to be part of a block. 

From the fact that $\seqreduce{\stmt}$ is well-typed we get that $\movertype(\stmt) \sqsubseteq \nonmover$, which implies that the mover types of actions in any execution path of statement $\stmt$ will match the pattern $\rightmover^*\nonmover?\leftmover^*$. 


\paragraph*{\bf{Commit transition of a block}} We define a commit transition of a block, when it exists, as: 
\begin{itemize}
    \item If the block contains a non-right-mover action transition, the commit transition is the first such transition within the block. 
    \item Otherwise, if the statement $\stmt$ executes completely or ends in a failure, the commit transition is the last transition of $\stmt$ or the failure transition, respectively, within the block.
\end{itemize}
All other transitions where seq-reduce is not applied are marked as commit transitions and each such transition is considered to be in a block of its own. 

The only case where a block does not have a commit transition is when it contains a statement from an application of a seq-reduce instance that (1) contains no non-right-mover action transition, (2) did not execute completely, and (3) did not end in a failure. 
A block that has a commit transition is called a \emph{committed} block; otherwise, it is called an \emph{uncommitted} block. 

\vspace{-2mm}
\paragraph*{\bf{Mover labeling for transitions in a block}} In every committed block, we assign the commit transition the non-mover label ($\nonmover$). 
All transitions before the commit transition are assigned the right mover label ($\rightmover$) and all transitions after the commit transition are assigned the left-mover label ($\leftmover$). 
All transitions in an uncommitted block are assigned the right-mover label. 
Given a transtion $t$ let $\labeling(t)$ denote the label assigned to it. 
Note that for any action transition $t_A$ of an action $A$, the mover type of $A$ is at least as precise as the label of $t_A$ (i.e., $\movertype(A) \sqsubseteq \labeling(t_A)$). 
For instance, if the mover type of an action was $\rightmover$, we would have labeled the transition as $\rightmover$ or $\nonmover$, but not $\leftmover$. 
After labeling, we find that each committed block has a labeling that matches the pattern $\rightmover^*\nonmover\leftmover^*$, with the single $\nonmover$ representing the commit transition, while every uncommitted block labeling matches the pattern $\rightmover^+$.

We define two predicates on transitions: $\committed(t_X)$ holds if transition $t_X$ belongs to a committed block, and $\uncommitted(t_X)$ holds if transition $t_X$ belongs to an uncommitted block. 
We can obtain the block of a transition $t_X$ using $\block(t_X)$ and its corresponding commit transition using $\commit{t_X}$.

\vspace{-2mm}
\paragraph*{\bf{Swapping Condition}}

Let $<_E$ denote the execution order over transitions. We define an arbitrary total order $<_O$ over the uncommitted blocks that remains fixed.

Two transitions $t_A$ and $t_B$ that occur consecutively in the execution are \emph{swappable} if and only if one of the following conditions holds:
\begin{enumerate}
    \item $\uncommitted(t_A) \land \committed(t_B)$
    \item $\committed(t_A) \land \committed(t_B) \land \commit{t_B} <_E \commit{t_A}$
    \item $\uncommitted(t_A) \land \uncommitted(t_B) \land \block(t_B) <_O \block(t_A)$. 
\end{enumerate}

\begin{lemma}
    If $t_A$ and $t_B$ are swappable, then either $\labeling(t_A) = \rightmover$ or $\labeling(t_B) =~\leftmover$.
\end{lemma}

\begin{proof}
If they satisfy conditions (1) or (3), then $t_A$ must have a right-mover label, since all uncommitted transitions have right-movers labels.
~\\
If they satisfy condition (2), then the only way $\commit{t_B} <_E \commit{t_A}$ conjunct of condition (2) can hold is if either $t_A$ has a right-mover label or $t_B$ has a left-mover label. 
    In any other case we violate condition (2). 

   \begin{enumerate}
    \item If we choose $\labeling(t_A)=\leftmover$ and $\labeling(t_B)=\rightmover$, then
    $\commit{t_A} <_E t_A$ and $t_B <_E \commit{t_B}$ by the definition of commit transitions. 
    Since every block has the form $\rightmover^{*}\nonmover\leftmover^{*}$, we obtain
    \[
        \commit{t_A} <_E t_A <_E t_B <_E \commit{t_B},
    \]
    which violates condition~(2).

    \item If we choose $\labeling(t_A)=\nonmover$ and $\labeling(t_B)=\rightmover$, then
    \[
        t_A=\commit{t_A} <_E t_B <_E \commit{t_B},
    \]
    again violating condition~(2).

    \item If we choose $\labeling(t_A)=\leftmover$ and $\labeling(t_B)=\nonmover$, then
    \[
        \commit{t_A} <_E t_A <_E \commit{t_B}=t_B,
    \]
    which also violates condition~(2).

    \item Finally, if we choose $\labeling(t_A)=\nonmover$ and $\labeling(t_B)=\nonmover$, then
    \[
        t_A=\commit{t_A} <_E \commit{t_B}=t_B,
    \]
    again violating condition~(2).
    \end{enumerate}
\hfill $\square$
\end{proof}

\begin{lemma}
    \label{lemma:seq-failure-preservation}
    Swapping preserves failures and it terminates.
\end{lemma}

\begin{proof}
Since swapping two transitions involving a local transition trivially preserves failures, we only need to show that swapping two action call transitions, $t_A$ and $t_B$, preserves failures. 
Let $A = (\blank , \blank, \gate_A, \blank, \MoverVar_A, \blank)$ and $B = ( \blank , \blank, \gate_B, \blank, \MoverVar_B, \blank)$. 
From our earlier observation---that if $t_A$ and $t_B$ are swappable, then either $\labeling(t_A) = \rightmover$ or $\labeling(t_B) = \leftmover$---we can conclude that either $\MoverVar_A \sqsubseteq \rightmover$ or $\MoverVar_B \sqsubseteq \leftmover$ holds.
Let $\ell_A= \inputbinding(t_A), \ell_B = \inputbinding(t_B)$ be their input bindings. 
~\\
\textbf{case 1}: $t_B$ is a failed transition
\[
(\globalStore_1, \ThreadPool_1) \xrightarrow{t_A} (\globalStore_2, \ThreadPool_2) \xrightarrow{t_B} \fail
\]
\begin{itemize}
    \item[] \textbf{case 1.1}: $\MoverVar_A \sqsubseteq \rightmover$ 
    \\We know from $\preservesSuccess(A, B)$ property~\condition{R1} of right movers that whenever gates of $A$ and $B$ hold in a state, every transition of $A$ from that state leads to a state where the gate of $B$ holds. 
    Using the contrapositive of this condition, we get the following: if, from $(\globalStore_1, \ThreadPool_1)$ there is a transition of $A$ to a state $(\globalStore_2, \ThreadPool_2)$ where the gate of $B$ does not hold, then either gate of $A$ does not hold in $(\globalStore_1, \ThreadPool_1)$ or the gate of $B$ does not hold in $(\globalStore_1, \ThreadPool_1)$. 
    Since, in our case, gate of $A$ holds in $(\globalStore_1, \ThreadPool_1)$, it must be that gate of $B$ does not hold in $(\globalStore_1, \ThreadPool_1)$. Therefore, we can eliminate $A$, and get a failure of $B$ from $(\globalStore_1, \ThreadPool_1)$.
    \item[] \textbf{case 1.2}: $\MoverVar_B \sqsubseteq \leftmover$
    \\In this case, we again use the $\preservesSuccess(A, B)$~\condition {L1} of left movers in contrapositive to eliminate $A$ and get a failure of $B$ (using the same reasoning as in the previous case).
\end{itemize}
~\\
\textbf{case 2}: $t_B$ is not a failed transition
\[(\globalStore_1, \ThreadPool_1) \xrightarrow{t_A} (\globalStore_2, \ThreadPool_2) \xrightarrow{t_B} (\globalStore_3, \ThreadPool_3)
\]
\begin{itemize}
    \item[] \textbf{case 2.1}: $(\globalStore_1, \localStore_A, \localStore_B)  \in wlp(A, \gate_B)$ and $ (\globalStore_1, \localStore_B, \localStore_A) \in wlp (B, \gate_A)$ 
    \\Since we have $\MoverVar_A \sqsubseteq \rightmover$ or $\MoverVar_B \sqsubseteq \leftmover$, we use the $\commutes(A,B)$ \condition{R2} or \condition{L3} to preserve the end state.

    \item[] \textbf{case 2.2}: $(\globalStore_1, \localStore_A, \localStore_B) \notin wlp(A, \gate_B)$ 
    \\We already have $\globalStore_1\cc\localStore_A \in \gate_A$. Therefore, we have a way to execute $A$ from $(\globalStore_1, \ThreadPool_1)$ and get a failure of $B$. 

    \item[] \textbf{case 2.3}: $(\globalStore_1, \localStore_B, \localStore_A) \notin wlp(B, \gate_A)$.
    \\If $\globalStore_1\cc\localStore_B \notin \gate_B$ then we execute $B$ from $(\globalStore_1, \ThreadPool_1)$ and get a failure. Otherwise, if $\globalStore_1\cc\localStore_B \in \gate_B$,
    then there exists a way to execute $B$ from $(\globalStore_1, \ThreadPool_1)$ and get a failure of $A$.

\end{itemize}

\paragraph*{Termination}
In all cases above, if a new failure is introduced, we obtain a strictly shorter failing execution and restart the swapping process on this new execution, since the sets of committed and uncommitted blocks may have changed.

Once the process stops restarting and the sets of committed and uncommitted blocks stabilize, each pair of transitions is swapped at most once. Indeed, after two transitions are swapped, the condition that enabled the swap no longer holds for that pair. \hfill $\square$
\end{proof}

\begin{lemma}
    \label{lemma:no-uncommitted-blocks}
    When no more swaps are enabled, there are no uncommitted blocks.
\end{lemma}

\begin{proof}
    First, we show that there are no uncommitted blocks. For the sake of contradiction, assume that there exists a transition from an uncommitted block. 
    Take the last uncommitted transition ($t_X$) in the execution.
    $t_X$ can't be a transition that failed, because a transition that failed belongs to a committed block. 
    This means there exists a transition that follows $t_X$ (call it $t_Y$). 
    $t_Y$ cannot be from an uncommitted block, otherwise $t_X$ would not be the last uncommitted transition. 
    Therefore, $t_Y$ must be committed and in this case a swap is enabled leading to a contradiction.
\hfill $\square$
\end{proof}

Now, we show that there are no blocks with interference.
We can now assume that the execution is made of only committed blocks, because we already showed that uncommitted blocks don't exist. 
We observe that committed blocks either look like $\rightmover^*\nonmover$ or $\rightmover^*\nonmover\leftmover^*$.
If a committed block B has interference, this means that there exists some transition, let's call this an interfering transition, of a different block in the execution that executes among the transitions of this block B. 
If the interfering transition executes before the commit transition of block B, then we say block B has a before-commit inteference. 
In any before-commit interference, you can find a right labeled transition of block B that immediately precedes the interfering transition. 
If the interfering transition executes after the commit transition of block B, then we say block B has an after-commit interference. 
In any after-commit inteference, you can find a left labeled transition that immediately follows the interfering transition. 

Let the order on blocks be defined by the order in which their commit transition occurs in the execution. 

\begin{lemma}
    \label{lemma:no-before-commit}
    When no more swaps are enabled, there is no block with before-commit interference.
\end{lemma}

\begin{proof}
    Let's assume for the sake of contradiction that there is some block with before-commit interference. 
    Pick the rightmost block B (block whose commit transition is the rightmost in the execution) with a before-commit interference.
    
    Let $t_R$ be the right labeled transition of a committed block B that immediately precedes the interfering transition $t_X$. 
    Note that $t_X$ also belongs to a committed block (using \autoref{lemma:no-uncommitted-blocks}).

    \textbf{case 1}: $t_X$ has label $\leftmover$ or $\nonmover$, then a swap is possible among $t_R$ and $t_X$ because $\commit{t_X} <_E \commit{t_R}$, contradicting the assumption that there are no more swaps enabled. 
        
    \textbf{case 2}: $t_X$ has label $\rightmover$ and $\commit{t_X} <_E \commit{t_R}$, \\then again a swap is possible among $t_R$ and $t_X$, contradicting the assumption that there are no more swaps enabled.
        
    \textbf{case 3}: $t_X$ has label $\rightmover$ and $\commit{t_R} <_E \commit{t_X}$, 
    \\implies we have $t_X <_E \commit{t_R} <_E \commit{t_X}$, which means we have a before-commit interference in $\block(t_X)$ where the interfering transition is $\commit{t_R}$ and this contradicts the assumption that block B is the rightmost with before-commit interference. 
\hfill $\square$
\end{proof}

\begin{lemma}
    \label{lemma:no-after-commit}
     When no more swaps are enabled, there is no block with after-commit interference.
\end{lemma}

\begin{proof}
    Assume there is some block with after-commit interference. 
    Pick the leftmost such block B (block whose commit transition is the leftmost in the execution) with an an after-commit interference. 

    Let $t_L$ be the left labeled transition that immediately follows the interfering transition $t_X$. 
    Note that $t_X$ also belongs to a committed block (using \autoref{lemma:no-uncommitted-blocks}).

    \textbf{case 1}: $t_X$ has label $\rightmover$ or $\nonmover$, \\ then a swap is possible because $\commit{t_L} <_E \commit{t_X}$, contradicting the assumption that there are no more swaps enabled.
  
    \textbf{case 2}: $t_X$ has label $\leftmover$ and $\commit{t_L} <_E \commit{t_X}$, \\ then again a swap is possible among $t_L$ and $t_X$ contradicting the assumption that there are no more swaps enabled.

    \textbf{case 3}: $t_X$ has label $\leftmover$ and $\commit{t_X} <_E \commit{t_L}$ \\ 
    implies we have $\commit{t_X} < \commit{t_L} <_E  t_X$, which means we have a before-commit interference in $\block(t_X)$ where the interfering transition is $\commit{t_L}$ and this contradicts the assumption that block B is the leftmost block with before-commit interference.
\hfill $\square$
\end{proof}

\begin{lemma}
    \label{lemma:no-uncommitted-no-interference-appendix}
    When no more swaps are enabled, there are no uncommitted blocks and there are no blocks with interference.
\end{lemma}

\begin{proof}
When no more swaps are enabled, from lemma~$\ref{lemma:no-uncommitted-blocks}$ we get that that there are no uncommitted blocks. 
Using lemma~$\ref{lemma:no-before-commit}$ and lemma~$\ref{lemma:no-after-commit}$ we don't have any committed blocks with interference. 
Since all blocks in the execution are committed (from \autoref{lemma:no-uncommitted-blocks}) we can conclude that all blocks are interference-free.
\hfill $\square$
\end{proof}

\paragraph*{\textbf{Satisfying completion requirement}}
Complete blocks are blocks corresponding to an instance of $\seqreduce{\stmt}$ in which either all transitions belonging to $\stmt$ have executed, or a transition belonging to $\stmt$ has failed. Otherwise, the block is called incomplete. 
We need to complete the incomplete blocks. Notice that any incomplete block requires only left-movers to complete it. 
If it required right-movers or a non-mover, it would have no commit transition—making the block uncommitted—and we have already eliminated all uncommitted blocks from the execution.
We pick the leftmost incomplete block and complete it the same way as we did for parallel reductions where we satisfy completion requirement for left-reduce case of $\parreduceempty$.
Once we complete a block, we still have an execution which is interference free and there are no uncommitted blocks, so we continue to complete the next leftmost incomplete block. 

\vspace{-2mm}
\paragraph*{\bf{Showing refinement property~\condition{P2}}} Given an execution $\pi$ of $\interprog$ from an initial configuration $(\globalStore, \ThreadPool)$ that ends in a final configuration $(\globalStore', \emptyset)$, we aim to show that there exists an execution $\pi'$ in $\reducedprog$ that starts from $(\globalStore, \ThreadPool)$ and ends in the same final configuration or fails.
\begin{align*}
    (\globalStore, \ThreadPool) \stepp[\prog]\transitive (\globalStore', \emptyset) \quad \implies \quad (\globalStore, \ThreadPool) \stepp[\prog']\transitive  (\globalStore', \emptyset) \quad \lor \quad (\globalStore, \ThreadPool) \stepp[\prog']\transitive  \fail
\end{align*}

We can establish this by applying the same swapping algorithm on $\pi$. If we don't introduce any failure, then we are guaranteed to preserve the end state $(\globalStore', \emptyset)$, as it follows from the $\commutes$ property of left and right movers. Otherwise, we end in a failure. 


\end{document}